\documentclass[10pt, final, journal, letterpaper, twocolumn]{IEEEtran}
\makeatletter
\def\ps@headings{%
	\def\@oddhead{\mbox{}\scriptsize\rightmark \hfil \thepage}%
	\def\@evenhead{\scriptsize\thepage \hfil \leftmark\mbox{}}%
	\def\@oddfoot{}%
	\def\@evenfoot{}}
\makeatother 

\IEEEoverridecommandlockouts
\usepackage{bbm}
\usepackage{amsfonts}
\usepackage[dvips]{graphicx}
\usepackage{times}
\usepackage{cite}
\usepackage{amsmath}
\usepackage{array}
\usepackage{amssymb}

\usepackage{stfloats}
\usepackage{slashbox}
\usepackage{graphicx}
\usepackage{footnote}
\usepackage{booktabs}
\usepackage{array}
\usepackage{algorithm}
\usepackage{subeqnarray}
\usepackage{cases}
\usepackage{threeparttable}
\usepackage{color}
\usepackage{hyperref}
\usepackage{epstopdf}
\usepackage{algpseudocode}
\usepackage{bm}
\usepackage{multirow}
\usepackage[labelformat=simple]{subcaption}
\usepackage{adjustbox}

\newtheorem{theorem}{Theorem}
\newtheorem{remark}{Remark}
\newtheorem{proposition}{Proposition}

\usepackage[absolute]{textpos}
\begin{document}
	
	\title{Applicable Regions of Spherical and Plane Wave Models for Extremely Large-Scale Array Communications}

	\author{\authorblockN{Renwang Li, Shu Sun, \IEEEmembership{Member,~IEEE}, and Meixia Tao, \IEEEmembership{Fellow,~IEEE}}\\
		\thanks{The authors are with the Department of Electronic Engineering, Shanghai Jiao Tong University, Shanghai, China (email:\{renwanglee, shusun, mxtao\}@sjtu.edu.cn).}
	}
	\maketitle
	
	\begin{abstract}
		Extremely large-scale array (XL-array) communications can significantly improve the spectral efficiency and spatial resolution, and has great potential in next-generation mobile communication networks. A crucial problem in XL-array communications is to determine the boundary of applicable regions of the plane wave model (PWM) and spherical wave model (SWM). In this paper, we propose new PWM/SWM demarcations for XL-arrays from the viewpoint of channel gain and rank. Four sets of results are derived for four different array setups. First, an equi-power line is derived for a point-to-uniform linear array (ULA) scenario, where an inflection point is found at  $\pm \frac{\pi}{6}$ central incident angles. Second, an equi-power surface is derived for a point-to-uniform planar array (UPA) scenario, and it is proved that $\cos^2(\phi) \cos^2(\varphi)=\frac{1}{2}$ is a dividing curve, where $\phi$ and $\varphi$ denote the elevation and azimuth angles, respectively. Third, an accurate and explicit expression of the equi-rank surface is obtained for a ULA-to-ULA scenario. Finally, an approximated expression of the equi-rank surface is obtained for a ULA-to-UPA scenario. With the obtained closed-form expressions, the equi-rank surface for any antenna structure and any angle can be well estimated.  {Furthermore, the effect of scatterers is also investigated, from which some insights are drawn.}
	\end{abstract}

	
	\section{Introduction}
	With the commercialization of the fifth-generation (5G) communication network, beyond 5G and the sixth-generation (6G) are now in sight \cite{zhang2020prospective, 8869705, you2021towards, wang2022vision}. 6G is expected to exploit new spectrum resources, including millimeter wave (mmWave) and Terahertz (THz) bands. With much shorter wavelengths of mmWave and THz bands as compared to the conventional microwave spectrum, more antennas can be integrated into a small space. Consequently, the number of antennas continues to increase, from a standard of 64 \cite{zhang2020prospective} to hundreds or thousands. Naturally, extremely large-scale multiple-input multiple-output (XL-MIMO) communication has gained great research interest in both academia and industry \cite{9187980}. XL-MIMO is also known as ultra-massive MIMO (UM-MIMO) \cite{akyildiz2016realizing}, extra-large scale massive MIMO \cite{9170651}, and extremely large aperture massive MIMO (xMaMIMO) \cite{8644126}. We use \emph{extremely large-scale array} (XL-array) \cite{9617121} throughout this paper for uniformity. XL-array communications are expected to significantly improve the transmission rate, spectral efficiency, and spatial resolution compared with existing systems. Therefore, XL-array has great potential in satellite, unmanned aerial vehicle (UAV), and reconfigurable  intelligent surface (RIS) assisted communication systems  \cite{9580418, 9672766, 8319526, zhang2021wireless, 9867922, 9724202, 9911191, s22145297}. 
	
	Generally, the electromagnetic radiation can be divided into far-field and near-field regions \cite{8736783}. In the far-field region, the angles at which the incident wave reaches each antenna can be considered the same. Thus, the corresponding channel can be modeled under the plane wave assumption. The plane wave model (PWM) only considers the angle and the number of antennas, and thus  is highly convenient for channel modeling and performance analysis. In the near-field  region, the channel is likely to be modeled under the spherical wave assumption, where the phase and the difference in traveling distance at different antennas cannot be ignored. The spherical wave model (SWM) includes the angles, the number of antennas, as well as the distance between the incident wave and each antenna, and it is accurate but complicated. 
	{Note that the antenna-channel interactions can be well described by spherical vector wave \cite{8453017, 8657705, 9091906}, which is different from the concept of SWM discussed herein.}
	Due to the large aperture of an XL-array, the near field  expands as compared with a conventional antenna array, thus the PWM may not be applicable any more in the proximity of the XL-array. Thus, it becomes crucial to determine the applicable regions of PWM and SWM for XL-array communications. A classical demarcation between near and far fields is the Rayleigh distance defined as $\frac{2 D^2}{\lambda}$ \cite{balanis2012advanced, balanis2015antenna, 7942128}, where $D$ and $\lambda$ denote the antenna/array aperture and signal wavelength, respectively. Specifically, the classical Rayleigh distance characterizes the minimum distance where the maximum phase difference across the whole antenna array is no greater than $\frac{\pi}{8}$ \cite{balanis2012advanced, balanis2015antenna}.  The Rayleigh distance is based on the phase difference which may not be the most appropriate evaluation metric in many practical systems. Hence, a new type of demarcations between the applicable regimes of PWM and SWM suitable for XL-arrays is urgently required.
	
	There are some previous efforts devoted towards the exploration of more accurate boundaries for PWM and SWM \cite{9617121, 1510955, 4155681, 4799060, 6800118, 7414041, liu2016channel, 7501567,7981398, lu2021does, cui2021near, bjornson2021primer}. The authors in \cite{1510955} find that the channel capacity for short-range MIMO systems under the PWM is lower than that under the SWM. A threshold distance below which the SWM is required for accurate performance estimation is obtained by using empirical fitting. The authors in \cite{4155681} propose a technique for realizing high-rank channel capacity in line-of-sight (LoS) MIMO transmission scenario by optimizing antenna placement. Based on the derived approximation expression of the largest eigenvalue of the LoS MIMO channels employing the uniform linear array (ULA) structure, the authors in \cite{4799060} investigate the distance where the ratio of the largest eigenvalue given by the SWM and that given by the PWM reaches a predefined threshold. The authors in \cite{6800118} obtain explicit expressions for some channel eigenvalues at certain discrete system settings under MIMO scenario, and propose an effective multiplexing distance where the channel can support a certain number of simultaneous spatial streams at a given signal-to-noise ratio (SNR) in a pure LoS environment. An analytical SWM for large linear arrays is proposed in \cite{7414041}, which is proved to be compatible with conventional PWM. The authors in \cite{liu2016channel} investigate the channel capacity under the point-to-point system and multi-user system. The performances under the SWM and the PWM are compared, and results show that the systems using the more realistic and accurate SWM can achieve a higher channel capacity than that under the PWM. 
	{The authors in \cite{7501567} and \cite{7981398} propose a space-alternating generalized expectation-maximization based localization algorithm, where the polarization, delay and Doppler effect are considered.}
	A closed-form expression of SNR with maximum ratio combining (MRC) is obtained in \cite{lu2021does} by considering the phase and amplitude variations among different antennas, which is extended in \cite{9617121}  by taking the projected aperture across array elements into consideration. They define a uniform-power distance where the power ratio between the weakest and strongest array element is no smaller than a certain threshold. The authors in \cite{cui2021near} define an effective Rayleigh distance according to the achievable array gain if the base station (BS) utilizes the far-field beamforming vector to serve a user located in the near field. The authors in \cite{bjornson2021primer} consider the near-field region of a passive antenna and an antenna array. They obtain an explicit expression of the distance where at most 3 dB propagation loss incurs. Recently,  the Rayleigh distance originally coined for the single-hop communication system has also been extended to two-hop  RIS-aided communication systems \cite{bjornson2021primer, 9538864, 9133126, zhang2021intelligent, cui2022near}.
	
	{Most of the above efforts focus on the phase change, although some existing works, e.g., \cite{9617121, cui2021near, bjornson2021primer}, propose PWM/SWM demarcations according to signal amplitudes solely or other different criteria, those are for  linear arrays, and a more comprehensive and accurate criterion is absent.} Therefore, we aim to explore the region where the SWM can be approximated by the PWM with high accuracy and under various circumstances, from the viewpoint of  { \textit{channel gain} and \textit{rank}} which are key evaluation metrics in wireless communication systems. We consider four scenarios, i.e., point-to-ULA, point-to-uniform planar array(UPA), ULA-to-ULA, and ULA-to-UPA, by taking into account the variations of both the signal amplitude and phase across array elements, as well as the incident angle.  {As mmWave or even THz techniques, along with XL-arrays,  may be employed in 6G, they are likely to be used for short/middle-range high-date-rate communication. Thus, the considered propagation  is LoS dominated or  with only  a limited number of scattered paths.}
	
	The main contributions of our work are summarized as follows:
	\begin{itemize}
		\item For the point-to-ULA scenario, we obtain an equi-power line where the ratio between the received power under the SWM and that under the PWM reaches a certain threshold. In addition, contrary to the intuition that the equi-power line would evolve smoothly over incident angles, it is discovered that an inflection point exists at the central incident angles of $\pm\frac{\pi}{6}$. When the absolute value of the central incident angle is smaller than $\frac{\pi}{6}$, the PWM always overestimates the received power compared with the SWM. Otherwise, the PWM first overestimates, then underestimates  the received power as the distance increases.
		\item For the point-to-UPA scenario, an equi-power surface is obtained. Similarly, it is found that $\cos^2(\phi) \cos^2(\varphi)=\frac{1}{2}$ is a dividing curve, where $\phi$ and $\varphi$ denote  the elevation and azimuth angles, respectively. When $\cos^2(\phi) \cos^2(\varphi)\geq \frac{1}{2}$, the PWM always overestimates the received power compared with the SWM. Otherwise, the PWM first overestimates, then underestimates  the received power as the distance increases.
		\item For the ULA-to-ULA scenario, we obtain an accurate and explicit expression of the equi-rank surface, where the effective rank \cite{7098875} of the channel under the SWM approximately equals that provided by the PWM. With the closed-form expression obtained, the equi-rank surface for any antenna structure and any angle can be well estimated.
		\item For the ULA-to-UPA scenario, an approximated expression of the equi-rank surface is obtained when the ULA is parallel to the UPA. The equi-rank surface for any antenna structure and any angle can be approximated using the derived expression.
		\item Multiple scatterers are also considered. For the point-to-ULA case, the phases of different paths cannot be cancelled out  as in the single-LoS-path case under  MRC, leading to a much larger threshold distance.  {The threshold distance enlarges as the number of scatterers increases both with and without the existence of an LoS path. For the ULA-to-ULA case, the obtained threshold distance approaches the analytical  approximation when the power of scattered paths is small.}
	\end{itemize}
	
	The rest of the paper is organized as follows. Section II investigates the equi-power line of a ULA. Section III discusses the equi-power surface of a UPA. Section IV explores the equi-rank surface of a ULA to a ULA. Section V investigates the equi-rank surface of a ULA to a UPA and Section VI concludes this paper.
	
	\begin{figure}[t]
		\begin{centering}
			\includegraphics[width=.4\textwidth]{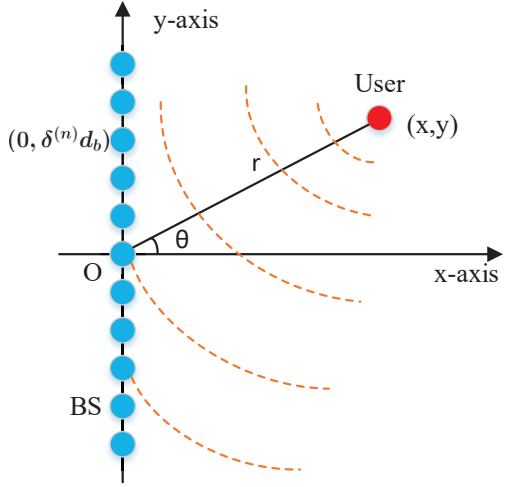}
			\caption{ {The system layout of a point to a ULA under the spherical wave model (SWM)}.}\label{fig_ula}
		\end{centering}
	\end{figure}
	\section{Equi-Power Line of a ULA} \label{sec_ula}
	In this section, we consider the equi-power line of an $N$-element ULA, as illustrated in Fig.~\ref{fig_ula}. Suppose the midpoint of the ULA lies at $(0,0)$ in the XY-plane. The coordinate of the $n$-th antenna is $(0, \delta^{(n)} d)$, where $\delta^{(n)} = n-\frac{N-1}{2}$ with $n=0,1,\ldots,N-1$, and $d$ denotes the antenna spacing. The array aperture is $D=(N-1)d\approx Nd$ for large $N$. A single-antenna user is located at $(x,y)$, with polar coordinate of $(r,\theta)$ where $r$ denotes the distance between the user and the center of the ULA, and $\theta$ is termed the central incident angle. Therefore, we have
	\begin{equation}
		\left\{\begin{array}{c}
			x= r\cos (\theta), \\
			y= r\sin (\theta).
		\end{array}\right.
	\end{equation}
	By considering the LoS path only, the exact channel $\mathbf{h} \in \mathbb{C}^{N\times 1}$ between the ULA and the user under the SWM is given by \cite{6800118, 7414041}
	\begin{equation} \label{cha_ula_exact}
		{\mathbf{h}} = \left[ \frac{\lambda}{4\pi r_1} e^{-j 2\pi \frac{r_1}{\lambda}}, \frac{\lambda}{4\pi r_2} e^{-j 2\pi \frac{r_2}{\lambda}},\ldots,\frac{\lambda}{4\pi r_N} e^{-j 2\pi \frac{r_N}{\lambda}} \right]^T,
	\end{equation}
	where $\lambda$ denotes the carrier wavelength, and $r_n$ refers to the distance between the $n$-th antenna and the user, satisfying
	\begin{equation}
		\begin{aligned}
			r_n &= \sqrt{x^2+ \left(y-\delta^{(n)} d \right)^2}\\
			&= \sqrt{r^2- 2r d \delta^{(n)} \sin (\theta) + \left(\delta^{(n)} d \right)^2}.\\
		\end{aligned}
	\end{equation}
	When the user lies in the far-field of the ULA, the channel in Eq. \eqref{cha_ula_exact} under the PWM can be approximated as
	\begin{equation}
		\overline{\mathbf{h}} = \frac{\lambda}{4\pi r} \left[  e^{-j 2\pi \frac{\overline{r}_1}{\lambda}},  e^{-j 2\pi \frac{\overline{r}_2}{\lambda}},\ldots, e^{-j 2\pi \frac{\overline{r}_N}{\lambda}} \right]^T,
	\end{equation}
	where $\overline{r}_n=r- d \delta^{(n)} \sin (\theta)$ due to the first order Taylor expansion $\sqrt{1+x} \approx 1+ \frac{1}{2} x$ which holds if the array length is much smaller than $r$. {When the beamformer of the ULA is adopted as $\mathbf{v}=\frac{\overline{\mathbf{h}}^H} {\|\overline{\mathbf{h}}\|}$, i.e., MRC, the maximum channel gain  {$G_\mathrm{far-field}$} can be achieved, which can be expressed as}
	\begin{equation}
		{G_\mathrm{far-field}}= (\mathbf{v} \overline{\mathbf{h}})^2 = \|\overline{\mathbf{h}}\|^2 = \left( \frac{\lambda}{4 \pi}\right)^2 \frac{N}{r^2}.
	\end{equation}
	However, the far-field approximation is inaccurate when the distance $r$ is small. Thereby, we would like to explore the boundary for PWM and SWM. With MRC, the maximum channel gain  {$G_\mathrm{near-field}$} can be achieved, which can be represented as
	\begin{equation}
		{G_\mathrm{near-field}}= \|\mathbf{h}\|^2 = \left( \frac{\lambda}{4 \pi}\right)^2 \sum \limits_{n=1}^N \frac{1}{r_n^2}.
	\end{equation}
	We are interested in the ratio of  {$G_\mathrm{near-field}$ to $G_\mathrm{far-field}$},  i.e.,
	\begin{equation}
		{\mu (r, \theta) \triangleq \frac{G_\mathrm{near-field}} {G_\mathrm{far-field}} }= \frac{r^2}{N} \sum \limits_{n=1}^N \frac{1}{r_n^2}.
	\end{equation}
	The coefficient $\mu(r, \theta)$ represents the normalized received power at position $(r,\theta)$. For some positions, the consequent coefficient $\mu(r, \theta)$ may be the same. It is obvious that when the distance $r$ goes to infinity, $\mu(r, \theta) \approx 1$ due to $r_n\approx r$.  We would like to study the boundary distance where the received power under the SWM is nearly the same as that under the PWM. The solutions $r(\theta)$ of $\mu(r, \theta)=\delta_\Delta$ represent the distances with the same normalized received power, where the threshold $\delta_\Delta$ is a constant that is close to one as for the PWM. Thus, we name the solutions $r(\theta)$ as \emph{equi-power line}. In the following, the closed-form expression of $\mu(r,\theta)$ is obtained through Theorem \ref{theorem_ula}, followed by the investigation of the equi-power line.
	
	\begin{theorem} \label{theorem_ula}
		In the point-to-ULA scenario, the normalized received power $\mu(r, \theta)$ has a closed-form expression as follows,
		\begin{equation} \label{expr_ula}
			\begin{aligned}
				\mu(r, \theta) =& \; \frac{r}{Nd\cos (\theta)} \left[ \arctan{\left(\frac{Nd}{2 r \cos (\theta)} + \tan{ (\theta) }\right) } \right. \\
				& + \left. \arctan{\left(\frac{Nd}{2 r \cos (\theta)} - \tan{ (\theta) } \right) } \right].
			\end{aligned}
		\end{equation}
	\end{theorem}
	\begin{proof}
		See Appendix \ref{appendix_theorem1}.
	\end{proof}
	It is evident from Eq. \eqref{expr_ula} that  $\mu(r, \theta)$ depends on the array length $Nd$ instead of the individual values of $N$ and $d$, indicating that $\mu(r, \theta)$ is identical for different ULAs with the same length. Moreover, $\mu(r, \theta)$ is independent of the wavelength.
	
	\begin{proposition}
		In the point-to-ULA scenario, $|\theta| = \frac{\pi}{6}$ is an inflection angle in the equi-power line. When $-\frac{\pi}{6} \leq \theta \leq \frac{\pi}{6}$, $\mu(r, \theta)$ is always smaller than one. When $ -\frac{\pi}{2} <\theta < -\frac{\pi}{6}$ or $ \frac{\pi}{6} <\theta < \frac{\pi}{2}$, $\mu(r, \theta)$ will first increase, exceed one, reach its peak, then decrease, and finally approach one as the distance $r$ increases, as shown in Fig.~\ref{fig_ula_line}. Assume $\mu(r,\theta)$ reaches its peak at the distance $r_1$,  where the PWM underestimates the received power the most severely. Define
		\begin{equation}
			r_{2} \triangleq Nd \sqrt{ \frac{\tan^2 (\theta)+1 + 2|\tan(\theta)| \sqrt{\tan^2(\theta)+1}}{4\cos^2(\theta) \left(\tan^2 (\theta)+ 1 \right) \left( 3 \tan^2(\theta)-1 \right)} },
		\end{equation}
		where $\frac{\partial^2 \mu(r, \theta)}{\partial r^2}|_{r=r_2}=0$, then $r_1$ is smaller than and close to $r_2$. 
	\end{proposition}
	\begin{proof}
		See Appendix \ref{appendix_prop1}.
	\end{proof}
	{For instance, when $\theta=\pi/3$, $N=127$, and $d=0.005$ m, we have $r_2=0.371$ m. Through simulation, we find $r_1\approx 0.274$ m.} It is observed that  $r_1$ is smaller than and close to $r_2$. 
	\begin{figure}[t]
		\begin{centering}
			\includegraphics[width=.45\textwidth]{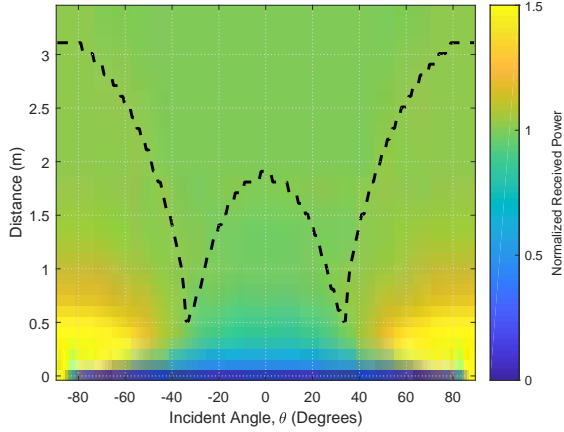}
			\caption{  {Point-to-ULA scenario: The normalized received power versus the central incident angle in the case of only an LoS path, where the dotted line denotes the equi-power line, $N=127$, and $\lambda=0.01$ m.}}\label{fig_ula_line}
		\end{centering}
	\end{figure}
	\begin{figure}[t]
		\begin{centering}
			\includegraphics[width=.4\textwidth]{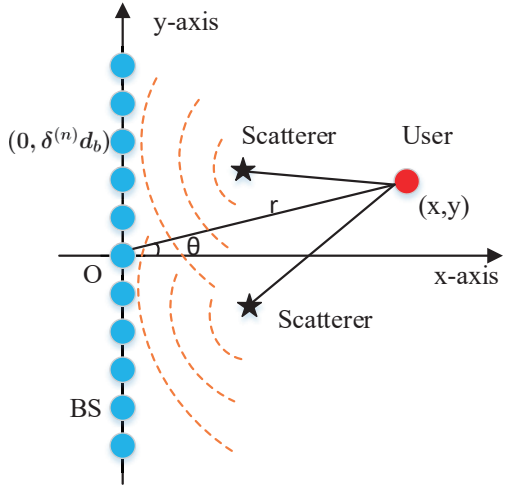}
			\caption{ {The system layout of a point to a ULA with multiple scatterers under the SWM.}}\label{fig_ula_scatters}
		\end{centering}
	\end{figure}
	\begin{figure}[t]
		\begin{centering}
			\includegraphics[width=.45\textwidth]{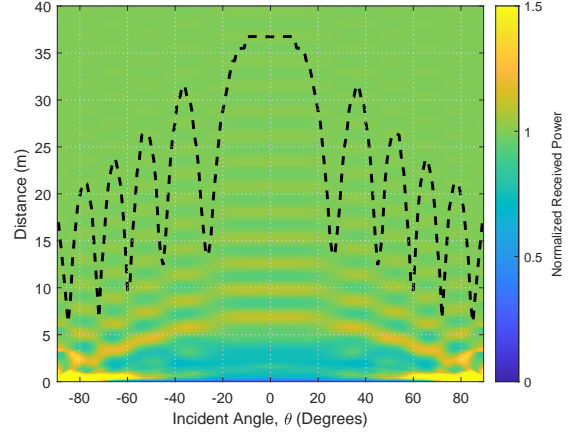}
			\caption{  {Point-to-ULA scenario: The threshold distance versus normalized received power versus the central incident angle in the case of two scatterers and an LoS path with $\zeta=5^\circ$, $|\alpha|=0.5$, $N=127$, and $\lambda=0.01$ m, where the dotted line denotes the equi-power line.}}\label{fig_ula_k0}
		\end{centering}
	\end{figure}
	\begin{figure}[t]
		\begin{centering}
			\includegraphics[width=.45\textwidth]{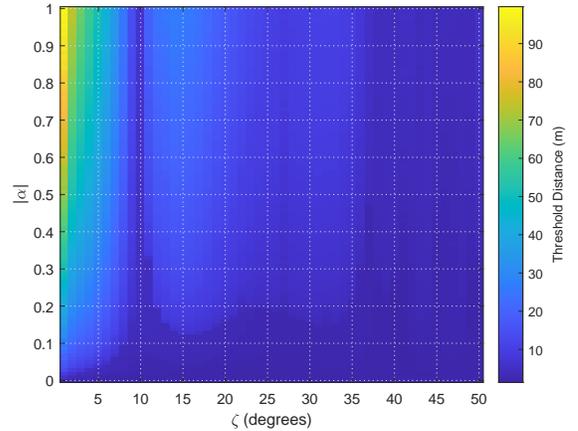}
			\caption{ {Point-to-ULA scenario: The threshold distance versus $|\alpha|$ and $\zeta$ with two scatterers and an LoS path, where $\theta=0^\circ$, $N=127$, and $\lambda=0.01$ m.}}\label{fig_ula_k-20}
		\end{centering}
	\end{figure}
	\begin{figure}[t]
		\begin{centering}
			\includegraphics[width=.45\textwidth]{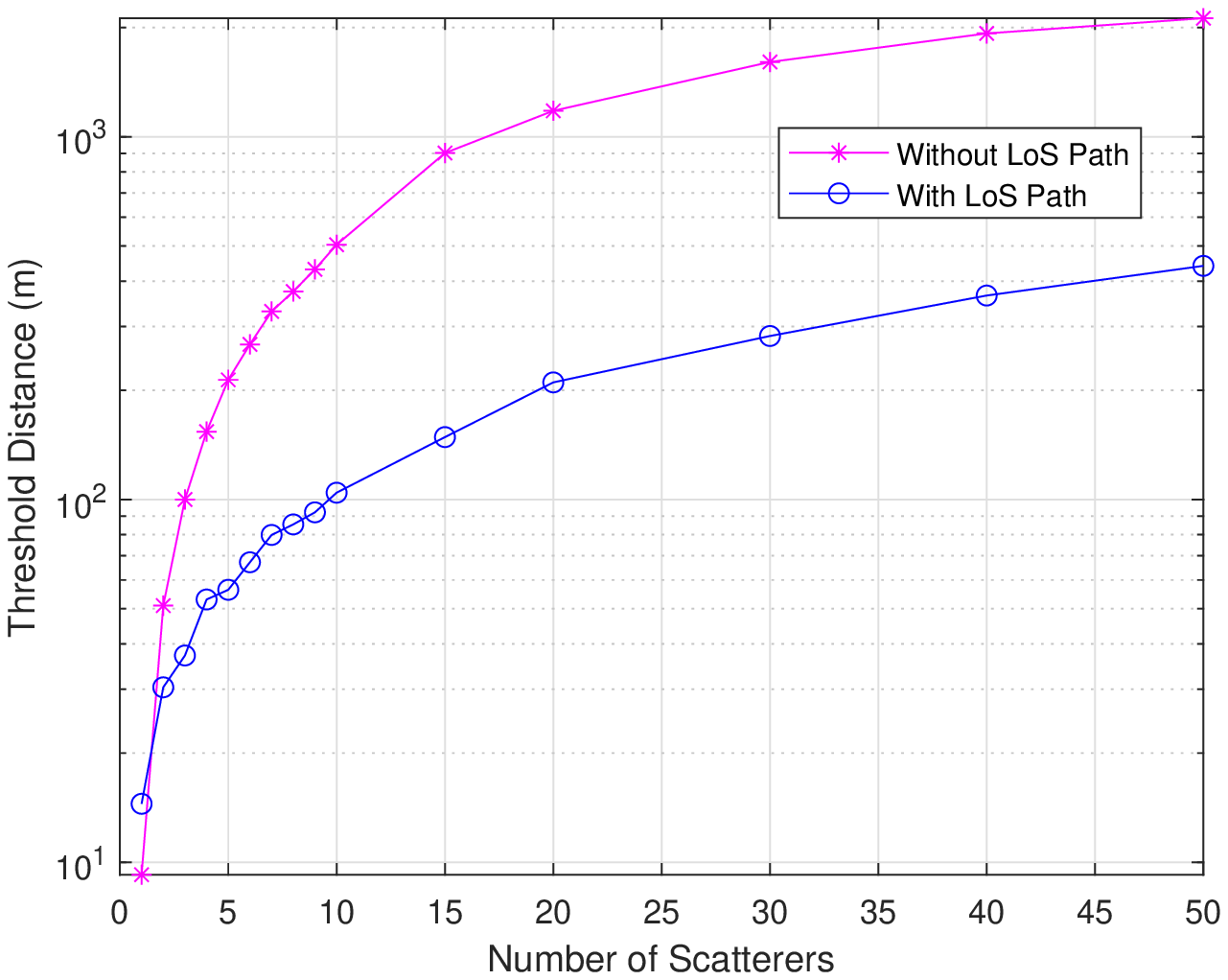}
			\caption{ {Point-to-ULA scenario: The threshold distance versus the number of scatterers where $\theta=0^\circ$, $N=127$, and $\lambda=0.01$ m.}}\label{fig_sca_num}
		\end{centering}
	\end{figure}
	
	Now, let us take a closer look at the equi-power line. When $\theta=0$, the coefficient $\mu(r, 0)$ can be simplified as
	\begin{equation} \label{mu_ula0}
		\mu(r,0) = \frac{2r}{Nd} \arctan{\left(\frac{Nd}{2r}\right)} .
	\end{equation}
	When $\delta_\Delta$ is set to 0.99, the corresponding equi-power line is $r(0)\approx 2.86Nd \approx 2.86D$. Fig.~\ref{fig_ula_line} illustrates the normalized received power and the equi-power line $r(\theta)$ for $\theta$ ranging from $-\frac{\pi}{2}$ to $\frac{\pi}{2}$, where $N=127$, $\lambda=0.01$ m, and $d=\frac{\lambda}{2}$. The dotted line denotes the equi-power line, where $\delta_\Delta$ is set to  0.99 when $\mu(r, \theta)$ is always smaller than one, and $\delta_\Delta$ is set to 1.01 otherwise. It can found that the equi-power line has a pseudo inflection point at $|\theta|>\frac{\pi}{6}$ due to the non-ideal threshold $\delta_\Delta$ which is not exactly one. When $-\frac{\pi}{6} \leq \theta \leq \frac{\pi}{6}$, the received power under the PWM is always a upper bound of that under the SWM. When $ \frac{\pi}{6} <|\theta| < \frac{\pi}{2}$, the received power corresponding to the PWM is first a upper bound, then a lower bound of that under the SWM as the distance increases. 
	\begin{table*}[t]
		\centering
		\caption{ {Comparison among different demarcations between applicable regions for SWM and PWM for the point-to-ULA scenario, where $\theta=0^\circ$, $\lambda=0.01$ m, $N=127$, and $D=0.63$ m.}}
		\begin{tabular}{cc}\hline
			\text{Different demarcations } & \text{Values (m)} \\ \hline
			\text { Classical Rayleigh distance   } & $\frac{2D^2}{\lambda}=79.38$  \\ \hline
			\text{Critical distance \cite{lu2021does}} & 0.70 \\ \hline		
			\text { Effective Rayleigh distance  \cite{cui2021near} } &  $ \epsilon \frac{2D^2}{\lambda}=29.13$ \\ \hline
			\text {The distance in \cite{bjornson2021primer}  } & $\frac{1}{10}\frac{2D^2}{\lambda}=7.94$ \\ \hline
			\text{Proposed equi-power line (LoS environment without scatterers)} & $2.86D=1.80$ \\ \hline
			{\text{Proposed equi-power line (LoS environment with 7 scatterers)}} &  {$79.87$} \\ \hline
			{\text{Proposed equi-power line (NLoS environment with 3 scatterers)}} &  {$100.04$} \\ \hline
		\end{tabular}
		\label{comp_ula}
	\end{table*}
	
	{So far, we have been focusing on the single LoS path between the user and the ULA. Now, let us consider the scenario with multiple scatterers that generate multipath as illustrated in Fig. \ref{fig_ula_scatters}. Let $\mathbf{h}^s \in \mathbb{C}^{N\times 1}$ denote the channel between the ULA and the user under the SWM, we have
		\begin{equation} \label{h_sca}
			{{h}^s_n} = \frac{\lambda}{4\pi r_n} e^{-j 2\pi \frac{r_n}{\lambda}}+ \sum \limits_{\ell=1}^L \frac{\alpha_\ell \lambda}{4\pi r_{n\ell} d_\ell} e^{-j 2\pi \frac{(r_{n\ell}+d_\ell)}{\lambda}} ,
		\end{equation}
		where $h^s_n$ is the $n$-th entry of $\mathbf{h}^s$, $L$ is the number of scatterers, $\alpha_\ell$ denotes the attenuation and phase shift caused by the $\ell$-th scatterer, $r_{n\ell}$ denotes the distance between the $\ell$-th scatterer and the $n$-th antenna, and $d_\ell$ denotes the distance between the user and the $\ell$-th scatterer. 
		The channel in \eqref{h_sca} under the PWM can be approximated by
		\begin{equation} \label{h_sca_pwm}
			{\overline{h}^s_n} = \frac{\lambda}{4\pi \overline{r}_n} e^{-j 2\pi \frac{\overline{r}_n}{\lambda}}+ \sum \limits_{\ell=1}^L \frac{\alpha_\ell \lambda}{4\pi \overline{r}_{n\ell} d_\ell} e^{-j 2\pi \frac{(\overline{r}_{n\ell}+d_\ell)}{\lambda}} ,
		\end{equation}
		where $\overline{r}_{n\ell}$ is the first order Taylor expansion of $r_{n\ell}$. It is challenging to analytically calculate the channel gain of Eq. \eqref{h_sca} and Eq. \eqref{h_sca_pwm} under the MRC. Therefore, simulations are carried out to evaluate the equi-power line. }
	
	Without loss of generality, we first consider  two scatterers with polar coordinates of $(r/2, \theta+ \zeta)$ and $(r/2, \theta-\zeta)$, respectively. We assume that the attenuation and phase shifts of the two scatterers are the same, i.e.,  $\alpha_1=\alpha_2$. The phase shifts are assumed to follow the uniform distribution between $-\pi$ and $\pi$, i.e., $\arg \{\alpha_1\}=\arg \{\alpha_2\} \sim \mathcal{U}(-\pi, \pi)$.
	After calculating the channel gain under SWM and PWM, we can obtain the coefficient $\mu(r,\theta)$ and the corresponding equi-power line, as shown in Fig. \ref{fig_ula_k0} and  {Fig. \ref{fig_ula_k-20}} where $\lambda=0.01$ m, $d=\frac{\lambda}{2}$, and $N=127$. The simulations are averaged over 100 independent realizations of scattering phase shifts.  When there is only a single LoS path between the user and ULA, the normalized received power is related to the distance, the array aperture, and the central incident angle, and is independent of phase. However, when multiple scatterers are considered, the phases of different paths cannot be cancelled out under the MRC as in the single-LoS-path case. Thus, the phases across the antennas will influence the normalized received power, leading to a much larger threshold distance, as shown in Fig. \ref{fig_ula_k0} when compared with Fig. \ref{fig_ula_line}.   {Fig. \ref{fig_ula_k-20} shows that the threshold distance decreases as $|\alpha|$ decreases since it becomes more akin to the single-LoS-path case. In addition, when the scatterers disperse, that is to say, $\zeta$ becomes large, the threshold distance is likely to decrease in general. The scattered signals strengthen each other in some places and weaken each other in others, as shown in Fig. \ref{fig_ula_k0}. When the scatterers are close together, they strengthen each other more. Thus, the threshold distance beyond which the normalized received power is always smaller than 1.01 becomes larger  in general. }
	
	{Then, we consider the circumstance where one or more scatterers are randomly distributed within a semicircle with radius $r$ with respect to the central element of the ULA, where $|\alpha_\ell|\sim \mathcal{U}(0, 1)$ and $\arg \{\alpha_\ell\} \sim \mathcal{U}(-\pi, \pi), \forall \ell$. Fig. \ref{fig_sca_num}  illustrates the threshold distance versus the number of scatterers with each data point averaged over 2000 independent realizations of the scatterer locations, and scattering attenuation and phase shifts , where $\theta=0^\circ$, $N=127$,  $\lambda=0.01$ m, and $d=\frac{\lambda}{2}$.  The curve `With LoS Path' indicates the LoS environment generated according to Eq. \eqref{h_sca}, while the curve `Without LoS Path' corresponds to the non-LoS (NLoS) environment with only scattered paths. It can be found that the threshold distance increases with the number of scatterers in both LoS and NLoS environment. For a given number of scatterers, the threshold distance in the NLoS environment is larger than that under the LoS condition when the number of scatterers exceeds two. The increased threshold distance with the number of scatterers may be ascribed to the fact that the scatterers close to the ULA are likely to be in the near field of the ULA where the SWM is applicable, whose effect needs to be overwritten by the distant scatterers such that the overall mixed wavefront approximates that formed by multiple plane waves.    As the number of scatterers increases, the number of close-in (with respect to the ULA) scatterers also increases, hence the distant scatterers need to be farther away to overcome the effect of the close-in scatterers,  so that the overall far-field effect can be presented.}
	
	\begin{figure}[t]
		\begin{centering}
			\includegraphics[width=.4\textwidth]{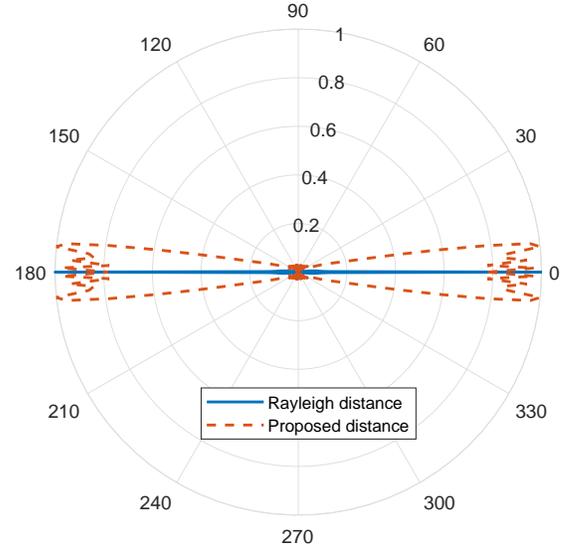}
			\caption{{ Point-to-ULA scenario: The beamforming pattern with the Rayleigh distance and the proposed distance in the case of only LoS path, where $\lambda=0.01$ m, and $N=127$.}}\label{fig_beampattern}
		\end{centering}
	\end{figure}
	Table \ref{comp_ula} compares different demarcations between applicable regions for SWM and PWM when $\theta=0^\circ$, $\lambda=0.01$ m, $N=127$, $D=0.63$ m, and $\epsilon=0.367$ \cite{cui2021near}. It can be found that the proposed equi-power line for a single LoS path is  smaller than most of the other boundaries. In fact, under MRC, the received power is affected by the variations of the signal amplitude, and is independent of the variations of the signal phase. Therefore, the classical Rayleigh distance overestimates the distance where the received power obtained by the SWM is nearly equal to that obtained by the PWM. The variations of the signal amplitude is also considered in \cite{lu2021does}, which cares about the power ratio between the weakest and strongest array elements. However, the received power by the whole array is considered in the proposed equi-power line, which results in different demarcations.  {The proposed equi-power line  is close to the classical Rayleigh distance under the LoS environment with 7 scatterers or the NLoS environment with 3 scatterers.} In addition, Fig. \ref{fig_beampattern} illustrates the beamforming pattern with the Rayleigh distance and the proposed distance in the scenario of Fig. \ref{fig_ula}, where $N=127$, $\theta=0^\circ$, $\lambda=0.01$ m, $d=\frac{\lambda}{2}$, and $r$ is set to the Rayleigh distance and the proposed distance, respectively.  It can be found that the beamforming pattern with the Rayleigh distance is sharp at $0^\circ$, which shows that it is indeed in the far field. However, the beamforming pattern with the proposed distance has large gains between $- 6^\circ$ and $6^\circ$.

	\begin{figure}[t]
		\begin{centering}
			\includegraphics[width=.45\textwidth]{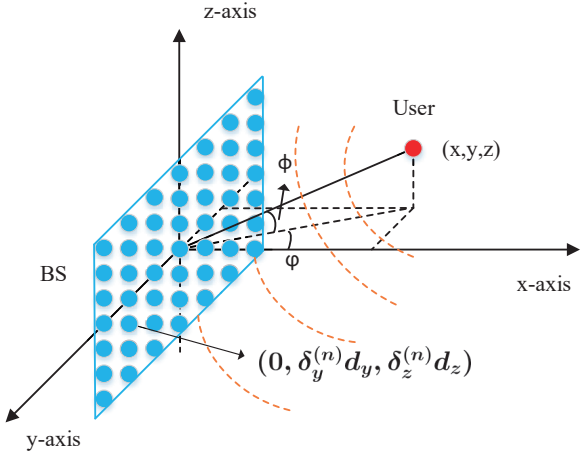}
			\caption{The system layout of a point to a UPA under the SWM.}\label{fig_upa}
		\end{centering}
	\end{figure}
	\section{Equi-Power Surface of a UPA} \label{sec_upa}
	We now consider the equi-power surface of an $N=N_y\times N_z$ UPA, as shown in Fig.~\ref{fig_upa}. Assume the UPA lies in the YZ-plane, and the midpoint of the UPA lies at $(0,0,0)$. The coordinate of the $n_y n_z$-th antenna is $(0, \delta_y^{(n)} d_y, \delta_z^{(n)} d_z)$, where $\delta_y^{(n)} = n_y -\frac{N_y-1}{2}$ with $n_y=0,1,\ldots,N_y-1$, $\delta_z^{(n)} =n_z -\frac{N_z-1}{2}$ with $n_z=0,1,\ldots,N_z-1$, $d_y$ denotes the antenna spacing along the Y direction, and $d_z$ denotes the antenna spacing along the Z direction. A single-antenna user is located at $(x,y,z)$, where its polar coordinate is $(r, \phi,\varphi)$. Thus, we have
	\begin{equation}
		\left\{\begin{array}{l}
			x= r\cos(\phi) \cos(\varphi), \\
			y= r\cos(\phi) \sin(\varphi), \\
			z= r\sin(\phi).
		\end{array}\right.
	\end{equation}
	\begin{figure*}
		\begin{equation} \label{upa_ori}
			\begin{aligned}
				\mu(r, \phi,\varphi) &=\frac{r^2}{N} \sum \limits_{n_y=1}^{N_y} \sum \limits_{n_z=1}^{N_z} \frac{1}{r_{n_y n_z}^2} \\
				&= \frac{r^2}{N} \sum \limits_{n_y=1}^{N_y} \sum \limits_{n_z=1}^{N_z}  \frac{1} {r^2- 2r \cos (\phi) \sin (\varphi) \delta_y^{(n)} d_y -2r \sin (\phi) \delta_z^{(n)} d_z + (\delta_y^{(n)} d_y)^2+ (\delta_z^{(n)}d_z)^2}  \\
				&= \frac{r^2}{N}  \sum \limits_{n_y=-\frac{N_y-1}{2}}^{ \frac{N_y-1}{2} } \sum \limits_{n_z=-\frac{N_z-1}{2}}^{\frac{N_z-1}{2}}  \frac{1}{r^2- 2r \cos (\phi) \sin (\varphi) n_y d_y -2r \sin (\phi) n_z d_z + n_y^2 d_y^2 + n_z^2 d_z^2 } \\
				&\overset{N \rightarrow \infty}{=}  \iint \limits_A \frac{r^2}{ r^2- 2r \cos (\phi) \sin (\varphi) N_y d_y {m_y}+ N_y^2 d_y^2 {m_y}^2  -2r \sin (\phi) N_z d_z {m_z} + N_z^2 d_z^2 {m_z}^2
				}  \mathrm{d} {m_y}  \mathrm{d} {m_z}.
			\end{aligned}
		\end{equation}
		\begin{equation} \label{mu_r00}
			\mu(r, 0,0)= \frac{2\pi r^{2}}{N_y d_y N_z d_z} \ln \left( \frac{N_z d_z \sqrt{\frac{N_y^2 d_y^2}{\pi r^2}+1} + N_y d_y \sqrt{\frac{N_z^2 d_z^2}{\pi r^2}+1}} {N_y d_y+ N_z d_z} \right).
		\end{equation}
		\begin{equation} \label{mu_rbeta}
			\mu_0 (r,\beta) =  \frac{\pi r^{2}}{2 D^{2}} \left\{ \ln \left[\frac{2 D^{2} \sqrt{\frac{D^{4}}{\pi^{2}}+\frac{\left(4 \beta -2\right) r^{2} D^{2}}{\pi}+r^{4}}+\left(\frac{2 D^{4}}{\pi}+\left(4 \beta -2\right) r^{2} D^{2}\right)}{2 D^{2} \sqrt{\frac{D^{4}}{\pi^{2}}+\frac{\left(4 \beta -2\right) r^{2} D^{2}}{\pi}+r^{4}}-\left(\frac{2 D^{4}}{\pi}+\left(4 \beta -2\right) r^{2} D^{2}\right)}\right]  +\ln \left[\frac{1- \beta }{ \beta }\right]\right\}.
		\end{equation}
	\end{figure*} 
	Analogous to the ULA structure in Section \ref{sec_ula}, the ratio of  {$G_\mathrm{near-field}$ to $G_\mathrm{far-field}$} under the UPA structure can be represented as
	\begin{equation}
		{\mu (r, \phi,\varphi) \triangleq \frac{G_\mathrm{near-field}} {G_\mathrm{far-field}}} = \frac{r^2}{N} \sum \limits_{n_y=1}^{N_y} \sum \limits_{n_z=1}^{N_z} \frac{1}{r_{n_y n_z}^2},
	\end{equation}
	where $r_{n_y n_z}$ denotes the distance between the $n_y n_z$-th antenna and the user, satisfying
	\begin{equation}
		\begin{aligned}
			r_{n_y n_z}^2 =&\; \left(r\cos(\phi) \cos(\varphi)\right)^2 + \left(r \sin(\phi)-\delta_z^{(n)} d_z \right)^2 \\
			& + \left(r\cos(\phi) \sin(\varphi)-\delta_y^{(n)} d_y \right)^2 \\
			=&\;  r^2- 2r \cos (\phi) \sin (\varphi) \delta_y^{(n)} d_y  \\
			& -2r \sin (\phi) \delta_z^{(n)} d_z + (\delta_y^{(n)} d_y)^2+ (\delta_z^{(n)}d_z)^2.
		\end{aligned}
	\end{equation}
	
	Similar to the derivation of the ULA structure in Appendix \ref{appendix_theorem1}, $\mu(r, \phi,\varphi)$ can be expressed as Eq. \eqref{upa_ori}, where $A=\{(m_y,m_z)|-\frac{1}{2}\leq m_y \leq \frac{1}{2}, -\frac{1}{2}\leq m_z \leq \frac{1}{2} \}$ is the integration area. The coefficient $\mu(r,\phi,\varphi)$ represents the normalized received power. The solutions $r(\phi,\varphi)$ of $\mu(r,\phi,\varphi)= \delta_\Delta$ denote approximately the same normalized received power at $(r,\phi, \varphi)$ between the SWM and PWM, and constitute a surface. Therefore, we name the solutions $r(\phi,\varphi)$ as \emph{equi-power surface}. In the following, we will first explore the closed-form expression of $\mu(r,\phi,\varphi)$, then investigate the equi-power surface.

	It is extremely difficult, if not impossible, to derive the closed-form expression of  $\mu(r, \phi,\varphi)$ in Eq. \eqref{upa_ori}. Fortunately, the closed-form expression of $\mu(r, \phi,\varphi)$ is available when considering the uniform circular planar array (UCPA) or uniform elliptical planar array (UEPA) structure, which will be shown next.
	\begin{theorem}
		In the single point to a UEPA scenario, when $\phi=\varphi=0$, the normalized received power $\mu(r,0,0)$ can be represented as Eq. \eqref{mu_r00}.
		
		In the single point to a UCPA scenario, assume $N_y d_y=N_z d_z \triangleq D$,  $\mu(r,\phi,\varphi)$ can be expressed as
		\begin{equation} \label{mu_rthetaphi}
			\mu(r, \phi,\varphi) =
			\begin{cases}
				\frac{\pi r^{2}}{D^2} \ln \left( \frac{D^2}{\pi r^2}+1 \right),  & \phi=\varphi=0, \\
				\mu_0 (r,\beta),
				& \phi \neq 0 \; \text{or}\; \varphi \neq 0,
			\end{cases}
		\end{equation}
		where $\mu_0(r,\beta)$ is shown in Eq. \eqref{mu_rbeta} and $\beta = \cos ^{2}(\phi) \cos ^{2}(\varphi)$.
	\end{theorem}
	\begin{proof}
		See Appendix \ref{appendix_theorem2}.
	\end{proof}
	It can be observed that both Eq. \eqref{mu_r00} and Eq. \eqref{mu_rthetaphi} are independent of the wavelength, and are determined by the array length $Nd$ rather than the individual values of $N$ and $d$.
	\begin{proposition} \label{pro_upa1}
		In the single point to a UCPA scenario, $\cos^2 (\phi) \cos^2 (\varphi) =\frac{1}{2}$ is a dividing curve. When $\cos^2 (\phi) \cos^2 (\varphi) \geq \frac{1}{2}$, $\mu(r,\phi,\varphi)$ is always smaller than one; when $\cos^2 (\phi) \cos^2 (\varphi) < \frac{1}{2}$, $\mu(r,\phi,\varphi)$ will first increase, and then decrease to approach one as the distance $r$ increases. Assume $\mu(r,\phi,\varphi)$ reaches its peak at the distance $r_1$, where the PWM underestimates the received power the most severely. Define
		\begin{equation}
			\begin{aligned}
				r_2 \triangleq D \sqrt{\frac{ 10+ \sqrt{100-9 (4\cos^2 (\phi) \cos^2 (\varphi)-2)^2 }}{9\pi (2-4\cos^2 (\phi) \cos^2 (\varphi))}},
			\end{aligned}
		\end{equation}
		which is close to the point where the second partial derivative of $\mu(r,\phi,\varphi)$ over $r$ is zero. Then, $r_1$ is smaller than and close to $r_2$. 
	\end{proposition}
	\begin{proof}
		See Appendix \ref{appendix_prop2}.
	\end{proof}
	\begin{figure}[t]
		\begin{centering}
			\includegraphics[width=.45\textwidth]{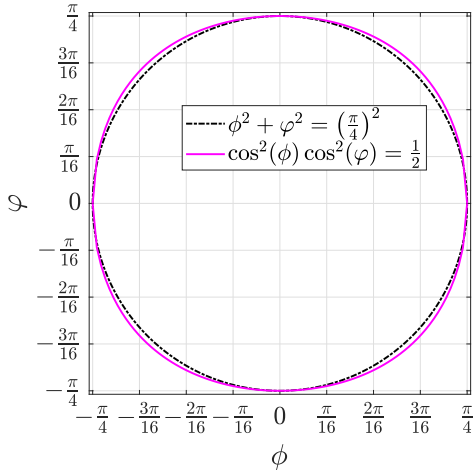}
			\caption{The dividing curve of a UCPA.} \label{upa_line}
		\end{centering}
	\end{figure}
	Fig.~\ref{upa_line} illustrates the dividing curve of a UCPA. When the angles $(\phi,\varphi)$ lie within the curve, the normalized received power $\mu(r,\phi,\varphi)$ is always smaller than one. It can be found that the dividing curve resembles a circle with a radius of $\frac{\pi}{4}$. In addition, for instance, when $D=1$ m, and $\cos^2 (\phi) \cos^2 (\varphi)=0.1$, we have $r_2=0.6442$ m. Through simulation, we find $r_1\approx 0.5257$ m. It is found that  $r_1$ is smaller than and close to $r_2$. 
	
	\begin{remark}
		It seems impossible to derive a closed-form expression for uniform square planar array (USPA) and uniform rectangular planar array (URPA) structures. Nevertheless, based upon extensive simulations, we find that $\cos^2 (\phi) \cos^2 (\varphi) =\frac{1}{2}$ is also the dividing curve of $\mu(r,\phi,\varphi)$ for the USPA structure. As for the URPA structure, let us focus on the region $\{(\phi,\varphi) | 0\leq \phi <\frac{\pi}{2}, 0\leq \varphi <\frac{\pi}{2} \}$, we have
		
		1) $\phi_0=\frac{\pi}{6}$ ($30^\circ$), and $\varphi_0 = \arccos{\sqrt{\frac{1}{2 \cos^2 (\pi/6)}}}$ ($\approx 35.26^\circ$) is the dividing curve;
		
		2) When $N_y d_y < N_z d_z$, $\cos^2 (\phi) \cos^2 (\varphi) =\frac{1}{2}$ is the dividing curve with the following condition: $0\leq \phi < \phi_0$ or $\varphi_0 < \varphi \leq \frac{\pi}{4}$;
		
		3) When $N_y d_y > N_z d_z$, $\cos^2 (\phi) \cos^2 (\varphi) =\frac{1}{2}$ is the dividing curve with the following condition: $\phi_0 < \phi \leq \frac{\pi}{4}$ or $0\leq \varphi <\varphi_0 $.
	\end{remark}
	
	\begin{figure}[t]
		\begin{centering}
			\includegraphics[width=.45\textwidth]{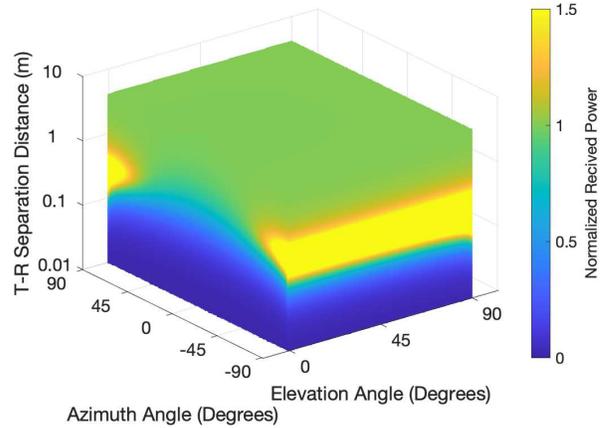}
			\caption{  {Point-to-UPA scenario: The normalized received power of a UCPA where $N_y = N_z = 127$, and $\lambda=0.01$ m.}}\label{upa_power}
		\end{centering}
	\end{figure}
	\begin{figure}[t]
		\begin{centering}
			\includegraphics[width=.45\textwidth]{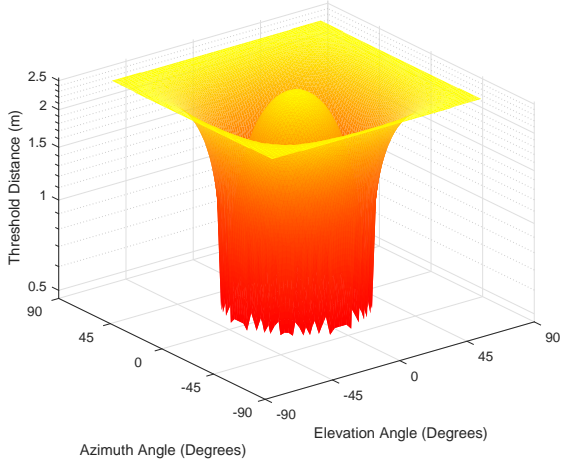}
			\caption{  {Point-to-UPA scenario: The equi-power surface of a UCPA where $N_y = N_z = 127$, and $\lambda=0.01$ m.}}\label{upa_surface}
		\end{centering}
	\end{figure}
	Considering the UCPA structure, when $\phi=\varphi=0$ and $\delta_\Delta=0.99$, the corresponding equi-power distance is $r(0,0)\approx 3.96 D$. Fig.~\ref{upa_power} illustrates the normalized received power at different angles, and Fig.~\ref{upa_surface} describes the equi-power surface $r(\phi,\varphi)$. Both  Fig.~\ref{upa_power} and Fig.~\ref{upa_surface} are obtained under the scenario of $N_y = N_z = 127$, $\lambda=0.01$ m, and $d_y=d_z=\frac{\lambda}{2}$. We set $\delta_\Delta$ to 0.99 when $\mu(r, \theta,\varphi)$ is always smaller than one, and  1.01 otherwise. The equi-power surface in Fig.~\ref{upa_surface} looks like an egg fallen into a hat, in which there is an inflection surface as explained in Proposition \ref{pro_upa1}.
	
	\begin{figure}[t]
		\begin{centering}
			\includegraphics[width=.35\textwidth]{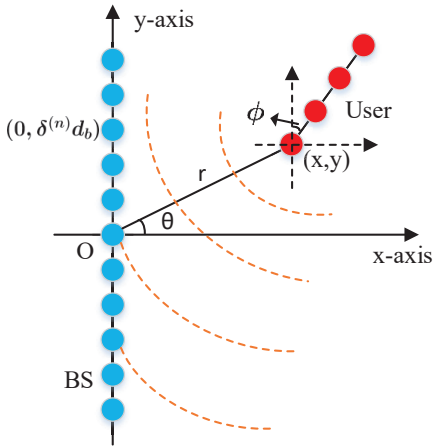}
			\caption{ The system layout of a ULA to a ULA under the SWM.}\label{fig_ula2ula}
		\end{centering}
	\end{figure}
	\section{Equi-Rank Surface of a ULA to a ULA} \label{sec_ula2ula}
	In this section, we focus on the equi-rank surface of a ULA to a ULA as shown in Fig.~\ref{fig_ula2ula}, where an $N$-antenna BS communicates with an $M$-antenna user. Suppose the midpoint of the ULA of the BS lies at $(0,0)$ in the XY-plane. The coordinate of the $n$-th antenna is $(0, \delta^{(n)} d_b)$, where $\delta^{(n)} = n-\frac{N-1}{2}$ with $n=0,1,\ldots,N-1$, and $d_b$ denotes the antenna spacing of the ULA at the BS. The coordinate of the first antenna of the user is $(x,y)$ with its polar coordinate $(r,\theta)$. The ULA of the user has an angle $\phi$ along with Y-axis. Thus, the coordinate of the $m$-th antenna of the user is $(x+ \delta^{(m)} d_u  \sin(\phi) ,y+ \delta^{(m)}d_u \cos(\phi))$, where $\delta^{(m)} = m$ with $m=0,1,\ldots,M-1$, and  $d_u$ is the antenna spacing of the ULA at the user. The channel  between the BS and the user under the SWM is denoted as $\mathbf{H}_s \in \mathbb{C}^{N\times M}$, where the $nm$-th entry of $\mathbf{H}_s$, i.e., $h_{nm}$, denotes the channel between the $n$-th antenna at the BS and the $m$-th antenna at the user, satisfying
	\begin{equation} \label{ula2ula_hmn}
		h_{nm} = \frac{\lambda}{4\pi r_{nm}} e^{ -j \frac{2\pi}{\lambda} r_{nm} }.
	\end{equation}
	Here, $r_{nm}$ denotes the distance between the $n$-th antenna at the BS and the  $m$-th antenna at the user, and it satisfies
	\begin{equation} \label{upa_rnm}
		\begin{aligned}
			r_{nm}^2 =&\; \left(r \cos(\theta) + \delta^{(m)}d_u \sin (\phi)\right)^2 \\
			& + \left(r \sin(\theta) + \delta^{(m)} d_u \cos(\phi)-\delta^{(n)} d_b\right)^2 .
		\end{aligned}
	\end{equation}
	We are interested in the matrix
	\begin{equation} \label{def_W}
		\mathbf{W}=
		\begin{cases}
			\mathbf{H}_s^H \mathbf{H}_s, & M<N \\
			\mathbf{H}_s \mathbf{H}_s^H, & M \geq N ,
		\end{cases}
	\end{equation}
	since the square of the singular values of $\mathbf{H}_s$ is equal to the eigenvalues of $\mathbf{W}$ \cite{4799060}. In the following, we will study the rank of $\mathbf{W}$.
	
	In order to investigate the rank of the channel, we introduce the effective rank of a non-all-zero matrix $\mathbf{A}\in \mathbb{C}^{M\times N}$  defined as \cite{7098875}
	\begin{equation}
		\operatorname{erank}(\mathbf{A}) = \exp \left( -\sum_i \frac{\sigma_i}{\|\mathbf{A}\|_*} \ln \left(\frac{\sigma_i}{\|\mathbf{A}\|_*}\right) \right),
	\end{equation}
	where $\sigma_i$ is the $i$-th largest singular value of $\mathbf{A}$, and $\|\mathbf{A}\|_*=\sum_i \sigma_i$ is the nuclear norm of $\mathbf{A}$, so that $\frac{\sigma_i}{\|\mathbf{A}\|_*}$ can be regarded as the normalized singular value of $\mathbf{A}$ \cite[Eq. (8)]{sun2018propagation}. As proven in \cite{7098875}, the effective rank of $\mathbf{A}$ is within $[1, \min(M,N)]$, and the effective rank is the lower bound of the rank of the matrix $\mathbf{A}$. With a higher value of the effective rank, the orthogonality of the columns/rows of the matrix is higher.
	
	\begin{figure}[htbp]
		\centering
		\vspace{-3mm}
		\subfloat[$\phi=0^\circ$ ]{\includegraphics[width=0.4\textwidth] {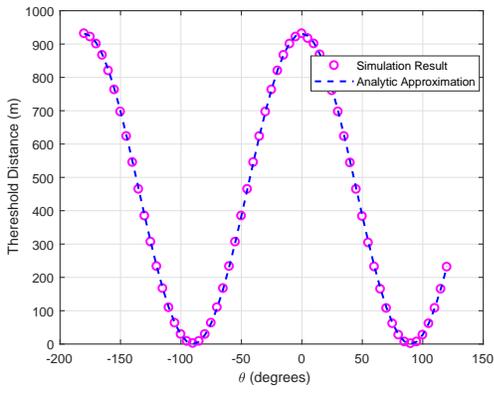}} \\
		\vspace{0mm}
		\subfloat[$\phi=30^\circ$] {\includegraphics[width=0.4\textwidth] {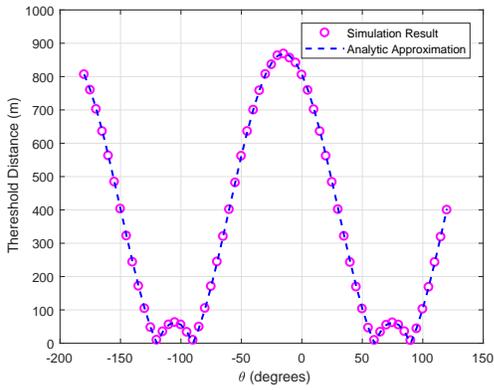}} \\
		\vspace{0mm}
		\subfloat[$\phi=50^\circ$] {\includegraphics[width=0.4\textwidth] {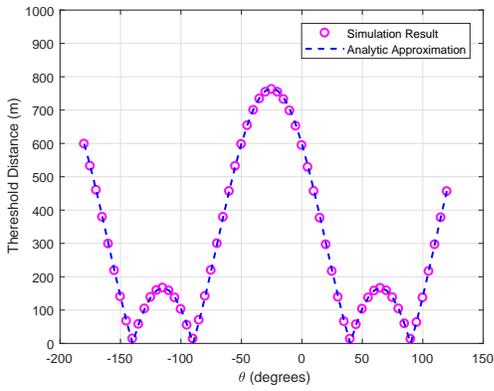}} \\
		\vspace{0mm}
		\subfloat[$\phi=60^\circ$] {\includegraphics[width=0.4\textwidth] {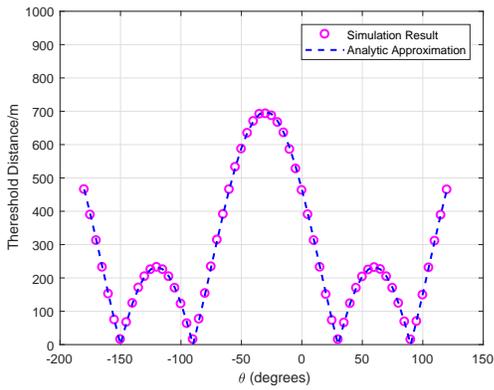}} \\
		\vspace{-1mm}
		\caption{ {ULA-to-ULA scenario: The threshold distance versus $\theta$ with fixed $\phi$, where $N=511$, $M=128$, and $\lambda=0.01$ m.}}
		\label{fig_ula2ula_line1}
	\end{figure}
	
	\begin{figure}[htbp]
		\centering
		\vspace{-3mm}
		\subfloat[$\theta=0^\circ$] {\includegraphics[width=0.4\textwidth] {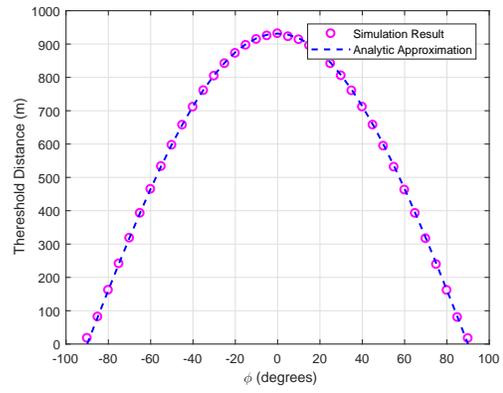}} \\
		\vspace{-0mm}
		\subfloat[$\theta=30^\circ$] {\includegraphics[width=0.4\textwidth] {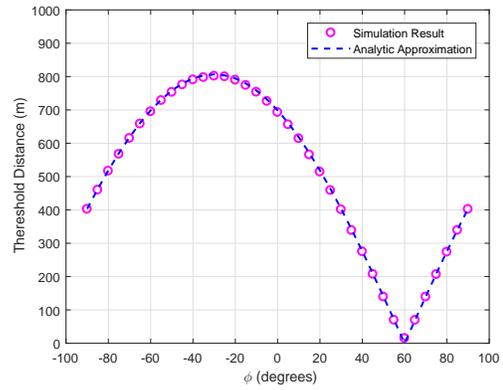}} \\
		\vspace{-0mm}
		\subfloat[$\theta=50^\circ$] {\includegraphics[width=0.4\textwidth] {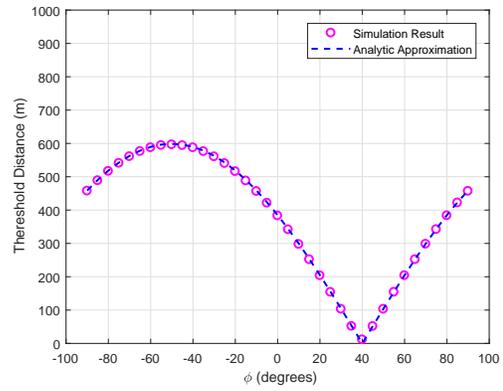}} \\
		\vspace{-0mm}
		\subfloat[$\theta=60^\circ$] {\includegraphics[width=0.4\textwidth] {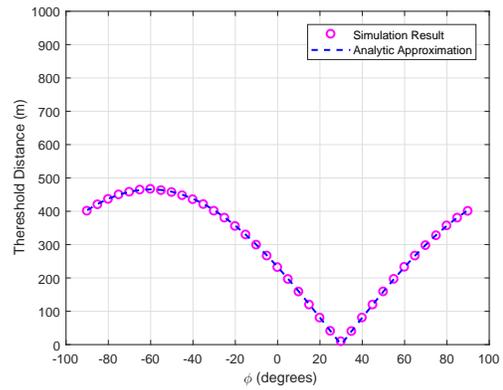}} \\
		\vspace{-1mm}
		\caption{ {ULA-to-ULA scenario:The threshold distance versus $\phi$ with fixed $\theta$, where $N=511$, $M=128$, and $\lambda=0.01$ m.}}
		\label{fig_ula2ula_line2}
	\end{figure}

	When the channel between the BS and the user is approximated by the PWM, its effective rank is one. However, the PWM is only valid when the distance is large. The solutions $r(\theta,\phi)$ of $\operatorname{erank}(\mathbf{W}) (r,\theta,\phi)= \delta_\Delta$ denote approximately the same effective rank against the PWM and constitute a surface. Hence, we name the threshold  $r(\theta,\phi)$ as \emph{equi-rank surface}. Here, $\delta_\Delta$ is a constant close to one. Since it is challenging to obtain the closed-form expression of the singular values of the matrix $\mathbf{W}$, we investigate the effective rank of the matrix $\mathbf{W}$ through simulations. In this section, the threshold $\delta_\Delta$ is set to 1.05.
	
	The investigation is divided into two parts. First, with fixed antenna numbers $N$ and $M$, the angles are sampled to find the relationships between the threshold distance  and the angles. Then, with fixed angles, the antenna numbers $N$ and $M$ are sampled to explore the relationships between the threshold distance and the antenna numbers. There are many interesting findings.
	
	Let us begin with the first part, as shown in Fig.~\ref{fig_ula2ula_line1} and Fig.~\ref{fig_ula2ula_line2}, where $N=511$, $M=128$, $\lambda=0.01$ m, and $d_b=d_u=\frac{\lambda} {2}$. Fig.~\ref{fig_ula2ula_line1} illustrates the threshold distance versus $\theta$ with fixed $\phi$, where $\theta$ is sampled every 5 degrees, $\phi=0^\circ$ in Fig.~\ref{fig_ula2ula_line1}(a), $\phi=30^\circ$  in Fig.~\ref{fig_ula2ula_line1}(b), $\phi=50^\circ$  in Fig.~\ref{fig_ula2ula_line1}(c), and $\phi=60^\circ$ in Fig.~\ref{fig_ula2ula_line1}(d). The markers and dotted curves denote the simulation results and analytic approximation, respectively. Note that the analytic expression will be expounded later. Denote the threshold distance when $\theta=\phi=0^\circ$ as $r_1$. There are some interesting findings: 1) When $\phi=0^\circ$, the simulation results in Fig.~\ref{fig_ula2ula_line1}(a) can be well approximated by $r_1 \cos^2 (\theta)$. 2) When $0<\phi< 90^\circ$, the simulation results have a peak and two sub-peaks when $-180^\circ\leq \theta\leq 120^\circ$. As $\phi$ increases, the value of the peak will decrease, and the values of the sub-peaks will increase. However, the summation of the peak and the sub-peak is always equal to $r_1$. 3) The peak arises at $\theta=-\frac{\phi}{2}$. The sub-peaks arise at $\theta=-\frac{\phi}{2} \pm 90^\circ$ both with a width of $\phi$ degrees.
	
	Fig.~\ref{fig_ula2ula_line2} displays the threshold distance versus $\phi$ with fixed $\theta$, where $\phi$ is sampled every 5 degrees, $\theta=0^\circ$ in Fig.~\ref{fig_ula2ula_line2}(a), $\theta=30^\circ$ in Fig.~\ref{fig_ula2ula_line2}(b), $\theta=50^\circ$ in Fig.~\ref{fig_ula2ula_line2}(c), and $\theta=60^\circ$ in Fig.~\ref{fig_ula2ula_line2}(d). A few remarks can be drawn: 1) When $\theta=0^\circ$, the simulation results in Fig.~\ref{fig_ula2ula_line2}(a) can be well approximated by $r_1 |\cos (\phi)|$, except for the points around $\phi=\pm 90^\circ$. 2) As $\theta$ increases, the maximum value decreases and can be well approximated by $r_1 \cos(\theta)$. 3) When the simulation results in Fig.~\ref{fig_ula2ula_line2} are normalized by its own maximum values, it can be well approximated by $|\cos(\phi+\theta)|$.
	
	Based on the above findings, we have the following proposition.
	\begin{proposition} \label{prop_ula2ula}
		If the threshold distance when $\theta_0=\phi_0=0^\circ$ is $r_1$, with fixed antenna structure and fixed wavelength, the equi-rank surface for any angles $\theta$ and $\phi$ can be approximated by
		\begin{equation} \label{ula2ula_eq1}
			r(\theta,\phi) = r_1 \left| \cos^2\left(\theta+\frac{1}{2}\phi\right) - \sin^2\left(\frac{1}{2}\phi\right) \right|.
		\end{equation}
	\end{proposition}
	
	Proposition \ref{prop_ula2ula} reveals the relationship between the equi-rank surface and the angles for a fixed antenna structure. The analytic approximation in Proposition \ref{prop_ula2ula} fits the simulation results well except for the scenario of $\phi=\pm 90^\circ-\theta$, as shown in Fig.~\ref{fig_ula2ula_line1} and Fig.~\ref{fig_ula2ula_line2}. When $\phi=\pm 90^\circ-\theta$, the ULA at the user forms a line to the origin. Then,  the actual threshold distance is slightly greater than zero. However, the analytic approximation in Proposition \ref{prop_ula2ula} is exactly zero. Thus, it causes errors in this case. Overall, the analytic approximation in Proposition \ref{prop_ula2ula} fits the simulation results  well.
	
	\begin{figure}[t]
		\begin{centering}
			\includegraphics[width=.45\textwidth]{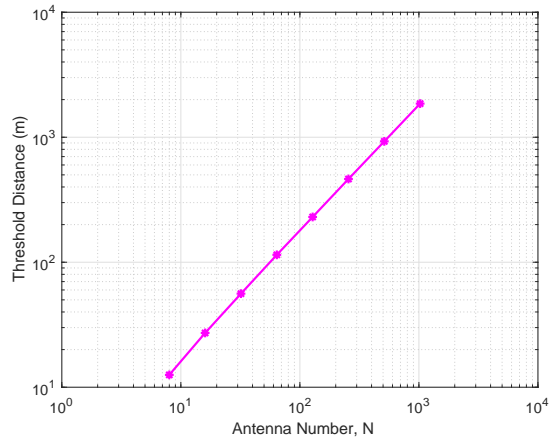}
			\caption{ {ULA-to-ULA scenario: The threshold distance versus the antenna number $N$, when $\theta=\phi=0^\circ$, $M=128$, and $\lambda=0.01$ m.}}\label{fig_ula2ula_line3}
		\end{centering}
	\end{figure}
	
	Now, let us focus on the second part. Fig.~\ref{fig_ula2ula_line3} shows the threshold distance versus the antenna number at the BS, where $\lambda=0.01$ m, and $d_b=d_u=\frac{\lambda} {2}$, $\theta=\phi=0^\circ$, $M=128$, and $N$ is taken from the set $\{ 8, 16, 32, 64, 128, 256, 512, 1024 \}$. It is worth noting that the threshold distance has a linear relationship with $N$, whether $N<M$ or $N\geq M$. Actually, the threshold distance has a linear relationship with the array aperture, not only the number of antennas, when the number of antennas of both the BS and user is larger than 6.  This linear relationship is likely due to the fact that when the number of antennas increases by a certain value, the curve of effective rank vs. distance translates to the right by a corresponding fixed amount. In addition, inspired by the classical Rayleigh distance $\frac{2D^2}{\lambda}$ where the Rayleigh distance and the wavelength  have an inverse relationship, we have the following proposition.
	\begin{proposition} \label{prop_ula2ula2}
		If the threshold distance for $\lambda_0$, $N_0$, $M_0$, $d_{b0}$, $d_{u0}$, and $\theta_0=\phi_0=0^\circ$ is $r_0$, then the threshold distance for any $\lambda$, $N$, $M$, $d_b$ and $d_u$, when $\theta_0=\phi_0=0^\circ$ can be approximated by
		\begin{equation} \label{ula2ula_eq2}
			r_1 =  \frac{N d_b M d_u}{N_0 d_{b0} M_0 d_{u0}} \frac{\lambda_0}{\lambda} r_0.
		\end{equation}
	\end{proposition}
	
	\begin{figure}[t]
		\begin{centering}
			\includegraphics[width=.45\textwidth]{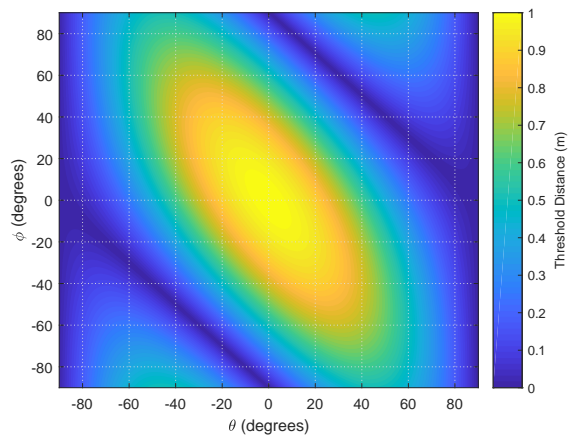}
			\caption{ The equi-rank surface of a ULA to a ULA.}\label{fig_ula2ula_surface}
		\end{centering}
	\end{figure}
	
	Proposition \ref{prop_ula2ula2} uncovers the relationships between the equi-rank surface and antenna structure when $\theta=\phi=0^\circ$. When $\lambda=0.01$ m, $d_{b0}=d_{u0}=\frac{\lambda}{2}$, and $N_0=M_0=100$, the threshold distance $r_0\approx 141.91$ m. Combining Proposition \ref{prop_ula2ula} and Proposition \ref{prop_ula2ula2}, the equi-rank surface for any antenna structure and any angle can be obtained. For instance, when $\lambda=0.005$ m, $d_{b}=d_{u}=\frac{\lambda}{2}$, $N=256$, $M=64$, $\theta=10^\circ$, and $\phi=30^\circ$, we have $r_1 = 116.25$ m according to Eq. \eqref{ula2ula_eq2}, the threshold distance $r = 87.70$ m according to Eq. \eqref{ula2ula_eq1}, while the result obtained by simulation is 87.60 m, showing that the analytic approximation fits the simulation results well. Fig.~\ref{fig_ula2ula_surface} illustrates the equi-rank surface of ULA to ULA when $r_1=1$ m in Eq. \eqref{ula2ula_eq1} for better understanding. 
	
	Actually, Proposition \ref{prop_ula2ula} and Proposition \ref{prop_ula2ula2} also work when $\delta_\Delta$ ($1<\delta_\Delta\leq 2$) is set to different values, with the difference of varied values of $r_0$ in Eq. \eqref{ula2ula_eq2}. Table \ref{eR_ula2ula} displays some values of $r_0$ in Eq. \eqref{ula2ula_eq2} versus $\delta_\Delta$, where $\lambda=0.01$ m, $d_{b0}=d_{u0}=\frac{\lambda}{2}$, and $N_0=M_0=100$. With smaller threshold $\delta_\Delta$, the distance $r_0$ is larger. The smaller the threshold $\delta_\Delta$ is, the closer the SWM is to the PWM, and the lower our error tolerance is.
	
	\begin{table}[t]
		\centering
		\caption{ {Values of $r_0$ in Eq. \eqref{ula2ula_eq2} versus $\delta_\Delta$, where $\lambda=0.01$ m, $d_{b0}=d_{u0}=\frac{\lambda}{2}$, and $N_0=M_0=100$.}}
		\begin{tabular}{cc}\hline
			\text{Threshold $\delta_\Delta$ } & $r_0$ (m) \\ \hline
			1.05 & 141.91  \\ \hline
			1.10 & 93.62 \\ \hline
			1.20 & 61.13 \\ \hline
			1.50 & 33.78 \\ \hline
			2.00 & 20.41 \\ \hline
		\end{tabular}
		\label{eR_ula2ula}
	\end{table}
	
	Besides, Table \ref{comp_ula2ula} compares different demarcations between applicable regions for SWM and PWM for the MIMO scenario, where $N=100$, $\lambda=0.01$ m, $d_u=d_b=\frac{\lambda}{2}$, and $\theta=\phi=0^\circ$. In Table \ref{comp_ula2ula}, {the classical Rayleigh distance is obtained according to $\frac{2(D_1+D_2)^2}{\lambda}$ \cite{cui2022near},} where $D_1$ and $D_2$ denote the array aperture at the BS and user, respectively; the distance in  \cite{4799060} is calculated by $\sqrt{\frac{(N^2-1)(M-1)^2 \pi^2}{6M(1-g)}} \frac{d_u d_b}{\lambda}$, where $g$ is set to 0.99. {The classical Rayleigh distance is larger than the other two demarcations since it is based on the phase difference.} In addition, the classical Rayleigh distance is independent of the angles, which is not practical. The authors in \cite{4799060} only consider the largest eigenvalue, while all eigenvalues are involved in our proposed equi-rank surface, and it turns out that the proposed equi-rank surface is larger than the distance in \cite{4799060}.
	\begin{table}[t]
		\centering
		\caption{ {Comparison among different demarcations between applicable regions for SWM and PWM for the ULA-to-ULA scenario, where $N=100$, $\lambda=0.01$ m, and $\theta=\phi=0^\circ$.}}
		\begin{tabular}{ccc}\hline
			\text{Different demarcations (m) } & $M=10$& $M=100$ \\ \hline
			\text { Classical Rayleigh distance   } & 60.50 & 200.00  \\ \hline
			\text {The distance in \cite{4799060}  } & 9.13 & 31.74 \\ \hline
			\text{Proposed equi-rank surface} & $14.19$ & 141.90\\ \hline
		\end{tabular}
		\label{comp_ula2ula}
	\end{table}
	
	\begin{figure}[t]
		\begin{centering}
			\includegraphics[width=.4\textwidth]{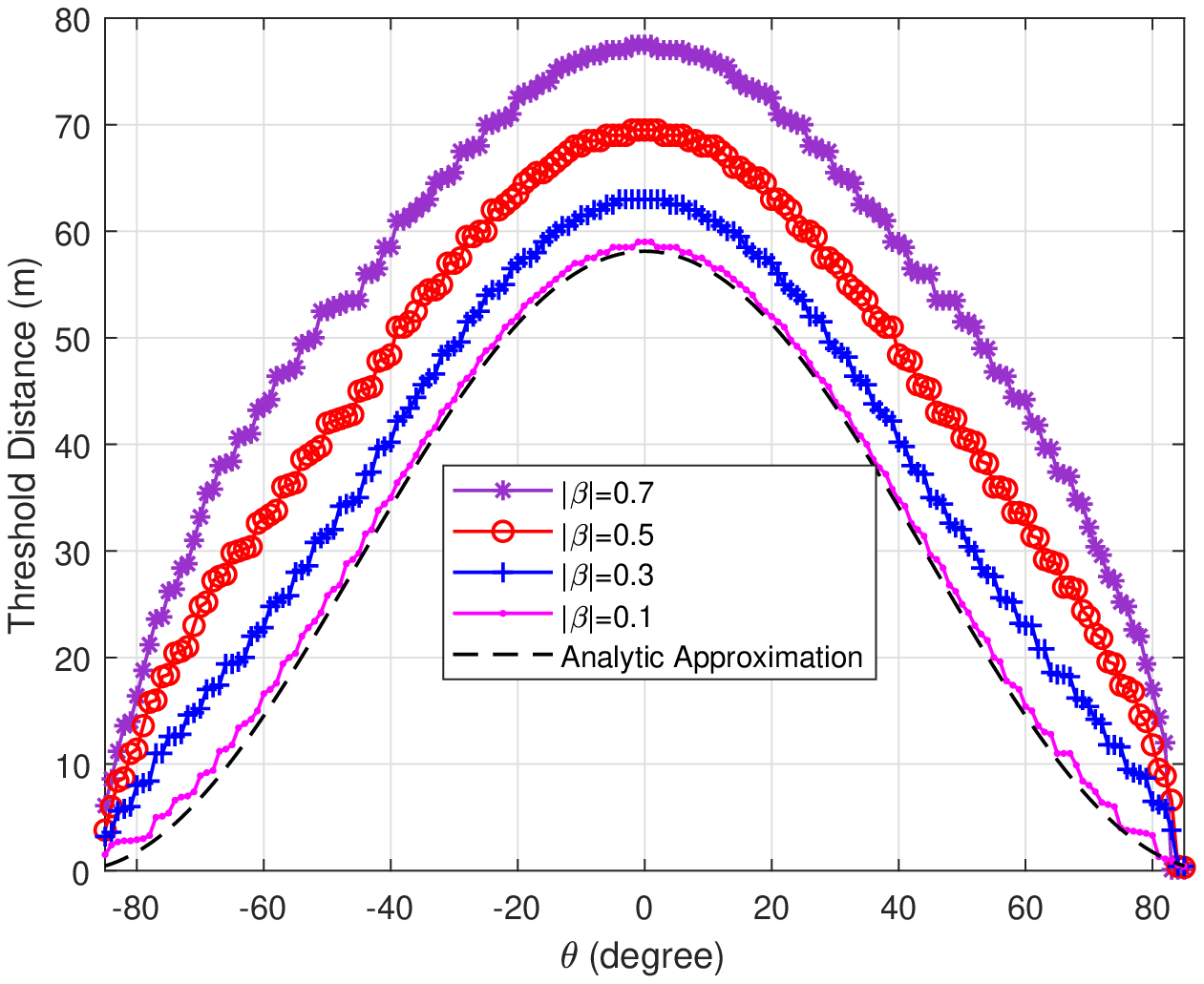}
			\caption{  {ULA-to-ULA scenario: The threshold distance versus $\theta$ with two scatterers, where $\phi=0^\circ$, $\zeta=5^\circ$, $N=128$, $M=32$, and $\lambda=0.01$ m.}}\label{phi0jiaodu5}
		\end{centering}
	\end{figure}
	\begin{figure}[t]
		\begin{centering}
			\includegraphics[width=.4\textwidth]{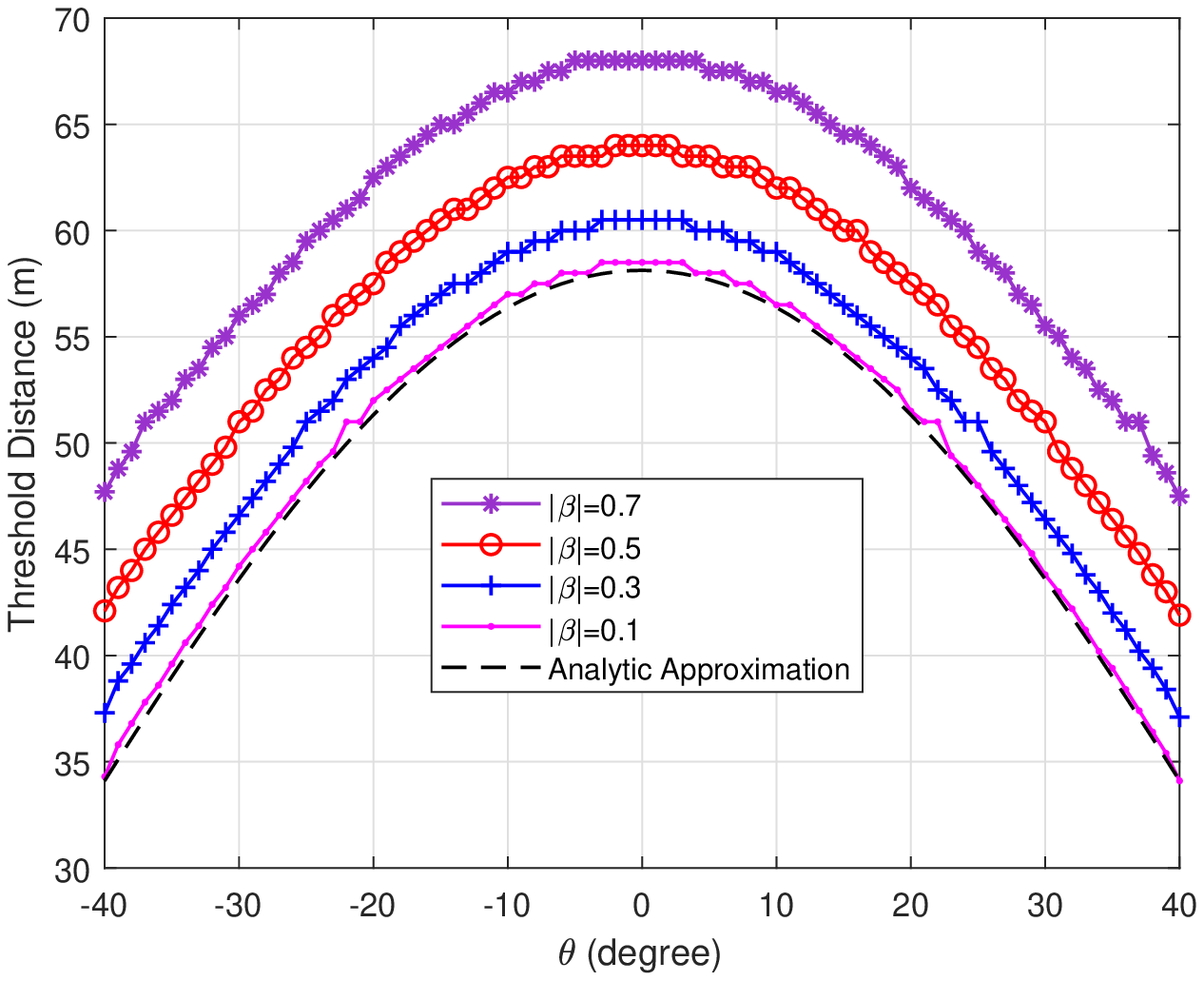}
			\caption{  {ULA-to-ULA scenario: The threshold distance versus $\theta$  with two scatterers, where $\phi=0^\circ$, $\zeta=50^\circ$, $N=128$, $M=32$, and $\lambda=0.01$ m.}}\label{phi0jiaodu50}
		\end{centering}
	\end{figure}
	\begin{figure}[t]
		\begin{centering}
			\includegraphics[width=.4\textwidth]{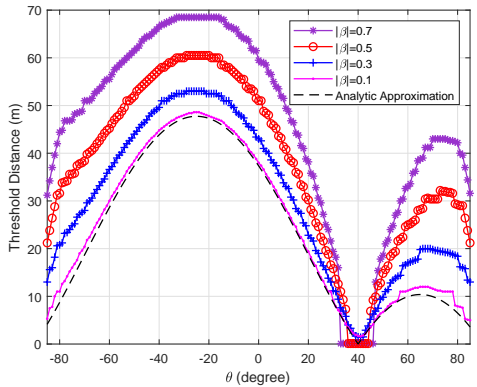}
			\caption{  {ULA-to-ULA scenario: The threshold distance versus $\theta$  with two scatterers, where $\phi=50^\circ$, $\zeta=5^\circ$, $N=128$, $M=32$, and $\lambda=0.01$ m.}}\label{phi50jiaodu5}
		\end{centering}
	\end{figure}
	{Similar to Section II, we consider multiple scatterers here. For simplicity and without loss of generality, suppose there are two scatterers with polar coordinates of $(r/2, \theta+ \zeta)$ and $(r/2, \theta-\zeta)$, respectively, where $\zeta$ is a constant angle. The channel in Eq. \eqref{ula2ula_hmn} can be represented as
		\begin{equation} \label{sca_ula2ula_hnm}
			h_{nm} = \frac{\lambda}{4\pi r_{nm}} e^{ -j \frac{2\pi}{\lambda} r_{nm} } + \sum \limits_{\ell=1}^2 \frac{\beta_\ell \lambda}{4\pi r_{n\ell} r_{\ell m}} e^{ -j \frac{2\pi}{\lambda} (r_{n\ell}+ r_{\ell m}) },
		\end{equation} 
		where $r_{n\ell}$ denotes the distance between the $n$-th antenna at the BS and the $\ell$-th scatterer, $r_{\ell m}$ denotes the distance between the $\ell$-th scatterer and the $m$-th antenna at the user, and $\beta_\ell$ denotes the attenuation and phase shift caused by the $\ell$-th scatterer. 
		For convenience, we assume $|\beta_1|=|\beta_2|=|\beta|$, and $\arg\{\beta_1\}= \arg \{\beta_2\}\sim \mathcal{U}(-\pi, \pi)$. Assume the user and scatterers are located in the far field of the ULA, we can obtain the effective rank of the channel, denoted as $\Xi$. Then, we define the distance where the effective rank of  the channel reaches $\Xi+0.05$ as the equi-rank surface. The user and scatterers are assumed to be located in the right-hand side of the ULA. The threshold distance versus $\theta$ is illustrated in Figs. \ref{phi0jiaodu5}-\ref{phi50jiaodu5} for various values of $\phi$ and $\zeta$, where $N=128$, $M=32$, $\lambda=0.01$ m, and $d_u=d_b=\frac{\lambda}{2}$. The simulations are averaged over 100 independent realizations of scattering phase shifts.   {It is found that when the power of the scatterers, i.e., $|\beta|$ is small, the obtained threshold distance is close to the analytical approximation.}

	\begin{figure}[t]
		\begin{centering}
			\includegraphics[width=.45\textwidth]{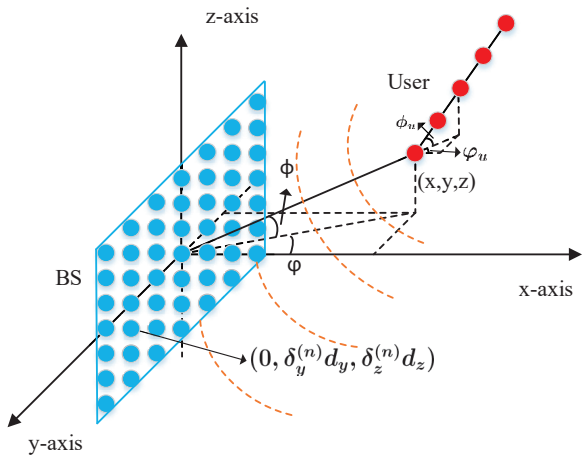}
			\caption{ The system layout of a ULA to a UPA under the SWM.}\label{fig_upa2ula}
		\end{centering}
	\end{figure}
	\section{Equi-Rank Surface of a ULA to a UPA}
	\begin{figure*}
		\begin{equation} \label{upa2ula_rn2}
			\begin{aligned}
				r_{n_y n_z,m}^2
				=& \left(r \cos(\phi) \cos(\varphi)+ \delta^{(m)} d_u \cos (\phi_u) \cos (\varphi_u) \right)^2 +\left(r \cos(\phi) \sin(\varphi)+ \delta^{(m)} d_u \cos \phi_u \sin \varphi_u - \delta_y^{(n)} d_y \right)^2 \\
				& +\left(r \sin(\phi)+ \delta^{(m)} d_u \sin \phi_u -\delta_z^{(n)} d_z \right)^2\\
			\end{aligned}
		\end{equation}
	\end{figure*}
	In this section, we focus on the equi-rank surface between an $M$-element ULA and an $N=N_y\times N_z$  UPA as shown in Fig.~\ref{fig_upa2ula}, where the UPA lies in the YZ-plane, and the midpoint of the UPA lies at $(0,0,0)$. The coordinate of the $n_y n_z$-th antenna is $(0, \delta_y^{(n)} d_y, \delta_z^{(n)} d_z)$, where $\delta_y^{(n)} = n_y -\frac{N_y-1}{2}$ with $n_y=0,1,\ldots,N_y-1$, $\delta_z^{(n)} =n_z -\frac{N_z-1}{2}$ with $n_z=0,1,\ldots,N_z-1$, $d_y$ denotes the antenna spacing along the Y direction, and $d_z$ denotes the antenna spacing along the Z direction. The coordinate of the first antenna of the user is $(x,y,z)$ with its polar coordinate $(r,\phi,\varphi)$. The azimuth angle and elevation angle of the ULA of the user are $\varphi_u$ and $\phi_u$, respectively. Thus, the coordinate of the $m$-th antenna of the user is $(x+ \delta^{(m)}d_u \cos(\phi_u) \cos(\varphi_u) ,y+ \delta^{(m)}d_u \cos(\phi_u) \sin(\varphi_u), z+ \delta^{(m)}d_u \sin(\phi_u))$, where $\delta^{(m)} = m$ with $m=0,1,\ldots,M-1$, and $d_u$ is the antenna spacing of the ULA at the user. The channel matrix between the BS and the user under the SWM is denoted as  $\mathbf{H}_s \in \mathbb{C}^{N\times M}$, where the $((n_z-1)N_y+n_y, m)$-th entry of $\mathbf{H}_s$, i.e., $h_{n_y n_z,m}$, denotes the channel between the $n_y n_z$-th antenna at the BS and the $m$-th antenna at the user, satisfying
	\begin{equation}
		h_{n_y n_z, m} = \frac{\lambda}{4\pi r_{n_y n_z, m}} e^{ -j \frac{2\pi}{\lambda} r_{n_y n_z, m} },
	\end{equation}
	where $r_{n_y n_z,m}$ denotes the distance between the $n_y n_z$-th antenna at the BS and the $m$-th antenna at the user, which can be expressed as Eq. \eqref{upa2ula_rn2}.
	\begin{figure}[t]
		\begin{centering}
			\includegraphics[width=.45\textwidth]{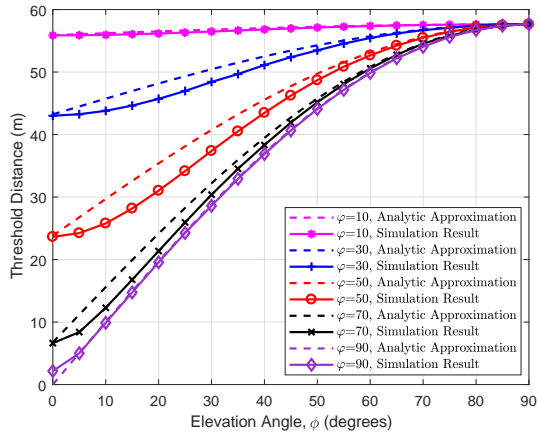}
			\caption{  {ULA-to-UPA scenario: The threshold distance versus the elevation angle $\phi$, when $N_y=N_z=127$, $M=32$, and $\lambda=0.01$ m.}}\label{fig_upa_line1}
		\end{centering}
	\end{figure}
	Similar to Eq. \eqref{def_W}, we are interested in the effective rank of the matrix $\mathbf{W}$. For brevity, we only consider the scenario where the ULA of the user is parallel to the UPA of the BS, which is the most typical setting in practical implementations. Specifically, we fix $\phi_u=0^\circ$ and $\varphi_u=90^\circ$. Hence, the ULA is parallel to the Y-axis. The solutions $r(\phi, \varphi)$ of $\operatorname{erank}(\mathbf{W}) (r,\phi,\varphi) = \delta_\Delta$ denote approximately the same effective rank against the PWM. Thereby, we name the threshold distance $r(\phi,\varphi)$ as \emph{equi-rank surface}. Similar to Section \ref{sec_ula2ula}, we explore the effective rank of the matrix $\mathbf{W}$ through two parts. The threshold $\delta_\Delta$ in this section is set to  1.05 as well.

	\begin{figure}[t]
		\begin{centering}
			\includegraphics[width=.45\textwidth]{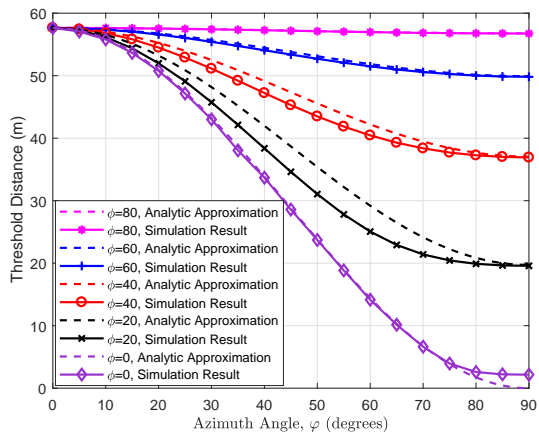}
			\caption{  {ULA-to-UPA scenario:The threshold distance versus the azimuth angle $\varphi$, when $N_y=N_z=127$, $M=32$, and $\lambda=0.01$ m.}}\label{fig_upa_line2}
		\end{centering}
	\end{figure}
	To begin with, let us focus on the first part to investigate the relationship between the threshold  distance and the angles for a fixed antenna structure, as shown in Fig.~\ref{fig_upa_line1} and Fig.~\ref{fig_upa_line2}, where $N_y=N_z=127$, $M=32$, $\lambda=0.01$ m, and $d_u=d_b=\frac{\lambda}{2}$. Fig.~\ref{fig_upa_line1} illustrates the threshold distance versus the elevation angles $\phi$ with a fixed azimuth angle $\varphi$, where the solid and dotted curves represent the simulation results and the analytic approximations, respectively. Note that the analytic expression will be explained later. Denote the threshold distance when $\phi=\varphi=0^\circ$ as $r_1$. There are some interesting observations: 1) When $\varphi=90^\circ$, the  simulation results fit the expression $r_1 \sin(\phi)$ well, except for the points around $\phi=0^\circ$. 2) As the elevation angle $\phi$ increases, all curves increase from $r_1 \cos^2(\varphi)$ to $r_1$.
	
	Fig.~\ref{fig_upa_line2} shows the threshold distance versus the azimuth angle $\varphi$ with a fixed elevation angle $\phi$. When $\phi=0^\circ$, the simulation results fit the expression $r_1 \cos^2 (\varphi)$ well apart from the points around $\varphi=90^\circ$. Furthermore, with the azimuth angle $\varphi$ increasing, all curves decrease from $r_1$ to $ r_1 \sin(\phi)$.
	
	Based on the above observations, we have the following proposition.
	\begin{proposition} \label{prop_upa2ula}
		For the USPA structure, if the threshold distance when $\phi_0=\varphi_0=0^\circ$ is $r_1$, with fixed antenna structure and fixed wavelength, the equi-rank surface for any angles $\phi$ and $\varphi$ is upper bounded by
		\begin{equation} \label{ula2upa_eq1}
			r(\phi,\varphi) \leq r_1 - r_1 \left( 1-|\sin(\phi)|\right) \left( 1-\cos^2(\varphi) \right).
		\end{equation}
	\end{proposition}
	\begin{figure}[t]
		\begin{centering}
			\includegraphics[width=.44\textwidth]{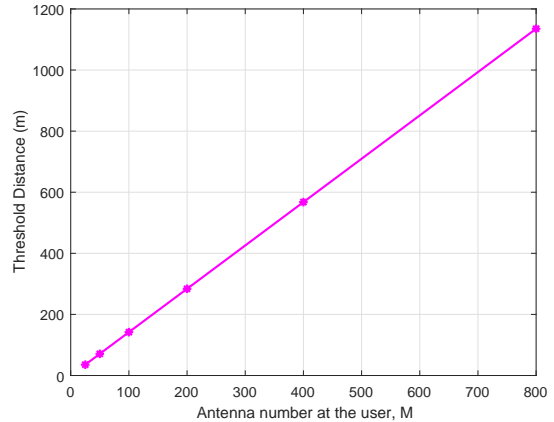}
			\caption{ {ULA-to-UPA scenario: The threshold distance versus the antenna number $M$, when $\phi=\varphi=0^\circ$, $N_y=100$, and $\lambda=0.01$ m.}}\label{fig_upa2ula_line3}
		\end{centering}
	\end{figure}
	
	Proposition \ref{prop_upa2ula} unveils the relationship between the equi-rank surface and the angles for a fixed antenna structure. Fig.~\ref{fig_upa_line1} and Fig.~\ref{fig_upa_line2} show that the analytic approximations in Proposition \ref{prop_upa2ula} fit the simulation results well when $\phi$ or $\varphi$ is close to $0^\circ$ or $90^\circ$. However, the analytic approximation is an upper bound of the simulation results for other angles. It may be caused by the antenna structure at the BS. We find that when $N_y\geq N_z$, the analytic approximation in Proposition \ref{prop_upa2ula} is the upper bound of the simulation results. However, when $N_y < N_z$, the simulation results may be larger than the analytic approximation. It is worth noting that we are only looking for a simple hence easy-to-understand analytic approximation here. More complex and precise approximations are left for future studies.
	\begin{figure}[t]
		\begin{centering}
			\includegraphics[width=.45\textwidth]{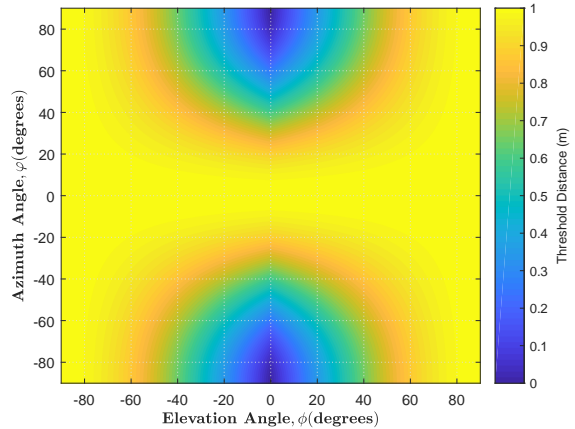}
			\caption{ The equi-rank surface of a USPA to a ULA.}\label{fig_upa_surface}
		\end{centering}
	\end{figure}
	
	Now, let us explore the relationship between the threshold distance and the antenna structure, as shown in Fig.~\ref{fig_upa2ula_line3}, where $\phi=\varphi=0^\circ$, $N_y=100$, and $\lambda=0.01$ m. The antenna number $M$ is selected from the set $\{ 25, 50, 100, 200, 400, 800\}$. It is remarkable that the threshold distance has a linear relationship with the antenna number $M$. Similar to Section \ref{sec_ula2ula}, actually, the threshold distance has a linear relationship with the antenna aperture, when $N_y>6$ and $M>6$. In addition, the threshold distance is independent of $N_z$ because the ULA at the user is parallel to Y-axis and $\phi=\varphi=0^o$. Consequently, we have the following proposition.
	\begin{proposition} \label{prop_upa2ula2}
		{If the threshold distance for $\lambda_0$, $N_{y0}$, $M_0$, $d_{y0}$, $d_{u0}$, and $\phi_0=\varphi_0=0^\circ$ is $r_0$, then the threshold distance for any $\lambda$, $N_y$, $M$, $d_b$ and $d_u$, when $\phi_0=\varphi_0=0^\circ$, can be approximated by}
		\begin{equation} \label{ula2upa_eq2}
			{r_1 = \frac{N_{y} d_{y} M d_u}{ N_{y0} d_{y0} M_0 d_{u0}} \frac{\lambda_0}{\lambda} r_0.}
		\end{equation}
	\end{proposition}
	Proposition \ref{prop_upa2ula2} reveals the relationship between the equi-rank surface and the antenna structure when $\phi=\varphi=0^\circ$. When $\lambda=0.01$ m, $d_{y0}=d_{u0}=\frac{\lambda}{2}$, and $N_{y0}=M_0=100$, the threshold distance $r_1\approx 141.91$ m. Combing Proposition \ref{prop_upa2ula} and Proposition \ref{prop_upa2ula2}, the equi-rank surface for any USPA structure at the BS, any ULA structure at the user, and any angles can be estimated. For instance, when $\lambda=0.005$ m, $d_{y}=d_z=d_{u}=\frac{\lambda}{2}$, $N_y=N_z=256$, $M=64$, and $\phi=\varphi=60^\circ$,  we have $r_0=116.25$ m based on Eq. \eqref{ula2upa_eq2}, $r=104.57$ m based on Eq. \eqref{ula2upa_eq1}, while the actual distance obtained by simulation results is $103.94$ m, demonstrating the accuracy of the analytic approximation. It is worth noting that similar to Section \ref{sec_ula2ula}, Proposition \ref{prop_upa2ula} and Proposition \ref{prop_upa2ula2} also work when $\delta_\Delta$ is set to different values. Finally, Fig.~\ref{fig_upa_surface} displays the equi-rank surface of a USPA to a ULA when $r_1=1$ m in Eq. \eqref{ula2upa_eq1} for more intuitive understanding.

	\section{Conclusion}
	This paper studied the demarcations between applicable regions for SWM and PWM for XL-array communication from the viewpoint of channel gain and rank. For the single-LoS-path scenario, closed-form expressions of the normalized received power for a single point to a ULA/UCPA were derived,and it was proved that $|\theta|=\frac{\pi}{6}$ and $\cos^2(\phi) \cos^2(\varphi)=\frac{1}{2}$ were the dividing angles and curves, respectively. As for the MIMO case, a closed-form expression of the equi-rank surface was obtained.  With the analytical expressions derived above, the equi-rank surface for any antenna structure and any angle could be well estimated.  
	
	{The impact of multiple scatterers on the SWM/PWM demarcation was also considered. For the point-to-ULA case, the threshold distance with the existence of LoS and scattered paths was larger than that in the single-LoS-path case since the phases of different paths could not be cancelled out under MRC. The threshold distance also increased with the number of scatterers in both LoS and NLoS environments, and the increase is more prominent for the NLoS environment. With fewer than 20 scattered paths in the LoS environment, the near field of a ULA can reach over 200 m, and this number is even larger for the NLoS environment, which are beyond the range of typical mmWave cells,  indicating the necessity of the consideration of SWM.  For the ULA-to-ULA case, the obtained threshold distance approached the analytical approximation when the power of scatterers was small.
		Future work may consider more accurate analytic approximations of the boundary surface for the ULA-to-UPA scenario as well as the effect of multiple scatterers in the scenarios containing UPAs.}
	
	\section*{Appendix}
	\subsection{Proof of Theorem 1} \label{appendix_theorem1}
	\begin{figure*}[htbp]
		\begin{equation} \label{prove_ula1}
			\begin{aligned}
				\mu(r, \theta) &=\frac{r^2}{N} \sum \limits_{n=1}^N \frac{1}{r_n^2} \\
				&= \frac{r^2}{N} \sum \limits_{n=1}^N \frac{1}{ r^2- 2r \sin (\theta) d \delta^{(n)} + \left(\delta^{(n)} d \right)^2 }  \\
				&= \frac{r^2}{N} \sum \limits_{n=0}^{N-1} \frac{1}{r^2- 2r \sin (\theta) d \left(n-\frac{N-1}{2} \right) + \left(n-\frac{N-1}{2} \right)^2 d^2}  \\
				&= \frac{r^2}{N} \sum \limits_{n=-\frac{N-1}{2}}^{\frac{N-1}{2}} \frac{1}{r^2 -2r \sin (\theta) nd+ n^2 d^2} \\
				&= r^2 \sum \limits_{m=-\frac{1}{2}+\frac{1}{2N}}^{\frac{1}{2}-\frac{1}{2N}} \frac{1}{N^2 d^2 m^2 -2r Nd\sin (\theta) m+r^2} \frac{1}{N}. \\
			\end{aligned}
		\end{equation}
		\begin{equation} \label{prove_ula2}
			\begin{aligned}
				\mu(r, \theta)
				&= r^2 \sum \limits_{m=-\frac{1}{2}+\frac{1}{2N}}^{\frac{1}{2}-\frac{1}{2N}} \frac{1}{N^2 d^2 m^2 -2r Nd\sin (\theta) m+r^2} \Delta_m \\
				&\overset{N\rightarrow \infty}{=} r^2\int_{-1/2}^{1/2} \frac{1}{N^2 d^2 m^2 -2r Nd\sin (\theta) m+r^2} \mathrm{d} m.
			\end{aligned}
		\end{equation}
		\begin{equation} \label{prove_ula_dao1}
			\begin{aligned}
				\frac{\partial \mu (r,\theta)}{\partial r} =&\; \dfrac{\left(-\frac{N d}{2\cos\left(\theta\right)\left(\left(\frac{N d}{2\cos\left(\theta\right)\,r}+\tan\left(\theta\right)\right)^2+1\right)r^2}-\frac{N d}{2\cos\left(\theta\right)\left(\left(\frac{N d}{2\cos\left(\theta\right)\,r}-\tan\left(\theta\right)\right)^2+1\right)r^2}\right)r}{N d\cos\left(\theta\right)}\\
				&+\dfrac{\arctan\left(\frac{N d}{2\cos\left(\theta\right)\,r}+\tan\left(\theta\right)\right)+\arctan\left(\frac{N d}{2\cos\left(\theta\right)\,r}-\tan\left(\theta\right)\right)}{N d\cos\left(\theta\right)}.
			\end{aligned}
		\end{equation}
		\begin{equation} \label{prove_ula_dao_f12}
			\begin{aligned}
				f_1(r,\theta) = &\;  16 \cos^4 (\theta) \left(\tan^2 (\theta)+ 1 \right) \left( 3 \tan^2(\theta)-1 \right) r^4  -8 N^2 d^2 \cos^2 (\theta) \left(\tan^2 (\theta)+ 1 \right) r^2-N^4 d^4. \\
				f_2(r,\theta) = &\; \left(\left(4\cos^2\left(\theta\right)\tan^2\left(\theta\right)+4\cos^2\left(\theta\right)\right)r^2-4N d\cos\left(\theta\right)\tan\left(\theta\right)\,r+N^2 d^2\right)^2\\
				& \times \left(\left(4\cos^2\left(\theta\right)\tan^2\left(\theta\right)+4\cos^2\left(\theta\right)\right)r^2+4N d\cos\left(\theta\right)\tan\left(\theta\right)\,r+N^2 d^2\right)^2.
			\end{aligned}
		\end{equation}
		\begin{equation} \label{prove_upa_mu1}
			\mu(r, 0,0)= r^{2} \int_{-\pi}^{\pi} \int_{0}^{1 / \sqrt{\pi}} \frac{\rho}{\left(N_y^2 d_y^2 \cos^2(\theta)+ N_z^2 d_z^2 \sin^2(\theta)\right) \rho^{2}+r^{2}} \mathrm{d} \rho\mathrm{d} \theta.
		\end{equation}
		\begin{equation} \label{prove_upa_mu2}
			\begin{aligned}
				\mu(r, 0,0)=&\; r^{2} \int_{0}^{1 / \sqrt{\pi}} \int_{-\pi}^{\pi} \frac{\rho} {r^{2}+ \rho^2 N_z^2 d_z^2+ \rho^{2} \left(N_y^2 d_y^2 - N_z^2 d_z^2 \right) \cos^2(\theta) } \mathrm{d} \theta \mathrm{d} \rho \\
				=&\; 4r^{2} \int_{0}^{1 / \sqrt{\pi}} \int_{0}^{\pi/2} \frac{\rho} {r^{2}+ \rho^2 N_z^2 d_z^2+ \rho^{2} \left(N_y^2 d_y^2 - N_z^2 d_z^2 \right) \cos^2(\theta) } \mathrm{d} \theta \mathrm{d} \rho.
			\end{aligned}
		\end{equation}
		\begin{equation} \label{prove_upa_mu3}
			\mu(r, 0,0)= \frac{2\pi r^{2}}{N_y d_y N_z d_z} \ln \left( \frac{N_z d_z \sqrt{\frac{N_y^2 d_y^2}{\pi r^2}+1} + N_y d_y \sqrt{\frac{N_z^2 d_z^2}{\pi r^2}+1}} {N_y d_y+ N_z d_z} \right).
		\end{equation}
	\end{figure*}
	
	\begin{figure*}	
		\begin{equation} \label{prove_ucpa1}
			\mu(r, \phi,\varphi) =r^{2} \int_{0}^{1 / \sqrt{\pi}} \int_{-\pi}^{\pi} \frac{\rho}{-2 r \cos (\phi) \sin (\varphi) D \rho \cos (\theta)-2 r \sin (\phi) D \rho \sin (\theta)+D^{2} \rho^{2}+r^{2}} \mathrm{d} \theta \mathrm{d} \rho.
		\end{equation}
		\begin{equation} \label{prove_ucpa2}
			\mu(r, \phi,\varphi)= 2 \pi r^{2} \int_{0}^{1 / \sqrt{\pi}} \frac{\rho}{\sqrt{D^{4} \rho^{4}+ \left(4 \cos ^{2}(\phi) \cos ^{2}(\varphi)-2\right) r^{2} D^{2} \rho^{2}+r^{4}}} \mathrm{d} \rho .
		\end{equation}
		\begin{equation} \label{prove_ucpa3}
			\mu (r,\beta) =  \frac{\pi r^{2}}{2 D^{2}} \left\{\ln \left[\frac{2 D^{2} \sqrt{\frac{D^{4}}{\pi^{2}}+\frac{\left(4 \beta -2\right) r^{2} D^{2}}{\pi}+r^{4}}+\left(\frac{2 D^{4}}{\pi}+\left(4 \beta -2\right) r^{2} D^{2}\right)}{2 D^{2} \sqrt{\frac{D^{4}}{\pi^{2}}+\frac{\left(4 \beta -2\right) r^{2} D^{2}}{\pi}+r^{4}}-\left(\frac{2 D^{4}}{\pi}+\left(4 \beta -2\right) r^{2} D^{2}\right)}\right]  +\ln \left[\frac{1- \beta }{ \beta }\right]\right\}.
		\end{equation}
	\end{figure*}
	The coefficient $\mu(r,\theta)$ can be calculated as Eq. \eqref{prove_ula1}.	Note that the operator $\sum \limits_{m=-\frac{1}{2}+\frac{1}{2N}}^{\frac{1}{2}-\frac{1}{2N}}$ calculates the summation with $m=-\frac{1}{2}+\frac{1}{2N}, -\frac{1}{2}+\frac{3}{2N},\ldots, \frac{1}{2}-\frac{1}{2N}$. Let $\Delta_m = \frac{1}{N}$, Eq. \eqref{prove_ula1} can be translated into Eq. \eqref{prove_ula2}. 	Since $\int \dfrac{1}{am^2-bm+c} \mathrm{d}m = \dfrac{2\arctan\left(\frac{2am-b}{\sqrt{4ac-b^2}}\right)}{\sqrt{4ac-b^2}} +Const$, where $Const$ denotes a constant number, by letting $a=N^2 d^2$, $b=2r Nd\sin \theta$ and $c=r^2$, we have
	\begin{equation} \label{ula_exp}
		\begin{aligned}
			\mu(r, \theta) =& \; \frac{r}{Nd\cos (\theta)} \left[ \arctan{\left(\frac{Nd}{2 r \cos (\theta)} + \tan{ (\theta) }\right) } \right. \\
			& + \left. \arctan{\left(\frac{Nd}{2 r \cos (\theta)} - \tan{ (\theta) } \right) } \right].
		\end{aligned}
	\end{equation}
	{It is worth noting that the closed-form expression in Eq.  \eqref{ula_exp} is accurate even when the antenna number $N$ is small, though it is derived when $N\rightarrow \infty$.}
	
	\subsection{Proof of Proposition 1} \label{appendix_prop1}
	Since $\mu(r, \theta)$ is an even function over the distance $r$. Here, we only consider $0\leq \theta< \frac{\pi}{2}$. We will prove the proposition by deriving its first partial derivative, shown in Eq. \eqref{prove_ula_dao1}, and the second partial derivative over the distance $r$,
	\begin{equation}
		\begin{aligned}
			\frac{\partial^2 \mu (r,\theta)}{\partial r^2} =&\; \frac{8 N^2 d^2 f_1(r,\theta)}{f_2(r,\theta)},
		\end{aligned}
	\end{equation}
	where $f_1(r,\theta)$ and $f_2(r,\theta)$ are shown in Eq. \eqref{prove_ula_dao_f12}.
	
	Let us focus on $f_1 (r, \theta)$. We consider the following two cases.
	
	1) When $\tan^2 (\theta)=\frac{1}{3}$, i.e., $\theta= \frac{\pi}{6}$, $f_1 =-8N^2 d^2- N^4 d^4<0$. Thereby, $\frac{\partial^2 \mu (r,\theta)}{\partial r^2}$ is always smaller than zero. Note that $\frac{\partial \mu (r,\theta)}{\partial r} \rightarrow 0$ as $r \rightarrow +\infty$. As a result, $\frac{\partial \mu (r,\theta)}{\partial r} \rightarrow 0$ is always larger than zero. In this case, the coefficient $\mu$ always increases as the distance increase, and is always smaller than one.
	
	2) When $\tan^2 (\theta)\neq \frac{1}{3}$, denote $c=-N^4 d^4$, $a= 16 \cos^4 (\theta) \left(\tan^2 (\theta)+ 1 \right) \left( 3 \tan^2(\theta)-1 \right)$,  and $b= -8 N^2 d^2 \cos^2 (\theta) \left(\tan^2 (\theta)+ 1 \right)$, the discriminant $\Delta = \left( 16 N^2 d^2 \cos^2(\theta) \tan (\theta)  \right)^2 \left(\tan^2(\theta) +1\right) >0$. Thus, the roots $r_{f_1}$ of $f_1(r, \theta)=0$ is given by
	\begin{equation}
		r_{f1} = N^2 d^2 \frac{\tan^2 (\theta)+1 \pm 2\tan(\theta) \sqrt{\tan^2(\theta)+1}}{4\cos^2(\theta) \left(\tan^2 (\theta)+ 1 \right) \left( 3 \tan^2(\theta)-1 \right)}.
	\end{equation}
	It can be found that when $0 \leq \theta< \frac{\pi}{6}$, $ \tan^2 (\theta)+1 - 2\tan(\theta) \sqrt{\tan^2(\theta)+1} >0$; when $\frac{\pi}{6} < \theta< \frac{\pi}{2}$, $ \tan^2 (\theta)+1 - 2\tan(\theta) \sqrt{\tan^2(\theta)+1} <0$.
	When $\tan^2(\theta) < \frac{1}{3}$, i.e., $0 \leq \theta < \frac{\pi}{6}$, the two numerators of $r_{f1}$ are larger than zero. The denominator of $r_{f1}$ is smaller than zero.
	Thus, the two roots $r_{f1}$ are smaller than zero. Therefore, $\frac{\partial^2 \mu (r,\theta)}{\partial r^2}$ is always smaller than zero.
	$\frac{\partial \mu(r,\theta)}{\partial r}$ is always bigger than 0. In this case, the coefficient $\mu$ is always smaller than one.
	
	When $\tan^2(\theta) > \frac{1}{3}$, i.e., $\frac{\pi}{6} < \theta<  \frac{\pi}{2}$, one numerators of $r_{f1}$ is larger than zero, and the other is smaller than zero. The denominator of $r_{f1}$ is larger than zero. Thus, one root is smaller than zero, and the other is larger than zero. Denote
	\begin{equation}
		r_{1} \triangleq Nd \sqrt{ \frac{\tan^2 (\theta)+1 + 2\tan(\theta) \sqrt{\tan^2(\theta)+1}}{4\cos^2(\theta) \left(\tan^2 (\theta)+ 1 \right) \left( 3 \tan^2(\theta)-1 \right)} },
	\end{equation}
	then when $0<r<r_1$, $\frac{\partial^2 \mu (r,\theta)}{\partial r^2} <0$; when $r> r_1$, $\frac{\partial^2 \mu (r,\theta)}{\partial r^2}>0$. In this case, the coefficient $\mu (r,\theta)$ will first increase, and then decrease as the distance $r$ increases. $r_1$ is the distance where the coefficient $\mu(r,\theta)$ drops the fastest. If $\mu (r,\theta)$ reaches its peak at the distance $r_2$, then $r_2$ is slightly smaller than $r_1$. When $ -\frac{\pi}{2} <\theta < -\frac{\pi}{6}$, we have
	\begin{equation}
		r_{1} = Nd \sqrt{ \frac{\tan^2 (\theta)+1 - 2\tan(\theta) \sqrt{\tan^2(\theta)+1}}{4\cos^2(\theta) \left(\tan^2 (\theta)+ 1 \right) \left( 3 \tan^2(\theta)-1 \right)} }.
	\end{equation}
	As such, $r_1$ can be expressed as
	\begin{equation}
		r_{1} = Nd \sqrt{ \frac{\tan^2 (\theta)+1 + 2|\tan(\theta)| \sqrt{\tan^2(\theta)+1}}{4\cos^2(\theta) \left(\tan^2 (\theta)+ 1 \right) \left( 3 \tan^2(\theta)-1 \right)} }.
	\end{equation}
	
	\subsection{Proof of Theorem 2} \label{appendix_theorem2}
	Note that area one for USPA is identical to that for a UCPA with a radius of $\frac{1}{\sqrt{\pi}}$. Thus, the integration area for USPA is $A=\{(\rho, \theta)|0\leq \rho \leq \frac{1}{\sqrt{\pi}}, -\pi \leq \theta \leq \pi \}$. Denote $m_y = \rho \cos(\theta)$, and $m_z=\rho \sin(\theta)$, let us consider the following two cases.
	
	1) When $\phi=\varphi=0$, considering uniform elliptical planar array structure, $\mu(r,0,0)$ can be expressed as Eq. \eqref{prove_upa_mu1}. Without loss of generality, assume $N_y d_y \geq N_z d_z$, Eq. \eqref{prove_upa_mu1} can be translated into Eq. \eqref{prove_upa_mu2}.
	Since $\int \frac{a}{b+c \cos^2 (\theta)} \mathrm{d} \theta = \dfrac{a\arctan\left(\frac{\sqrt{b}\tan\left(\theta\right)} {\sqrt{c+b}}\right)}{\sqrt{b}\sqrt{c+b}} +Const$, let $a=\rho$, $b=r^{2}+ \rho^2 N_z^2 d_z^2$, and $c=\rho^{2} \left(N_y^2 d_y^2 - N_z^2 d_z^2 \right)$, yielding
	\begin{equation}
		\mu(r, 0,0)= \int_{0}^{1 / \sqrt{\pi}} \frac{ 2\pi r^{2} \rho}{ \sqrt{r^2+ \rho^2 N_z^2 d_z^2} \sqrt{r^2 + \rho^2 N_y^2 d_y^2}} \mathrm{d} \rho.
	\end{equation}
	Since $\int \dfrac{ \rho }{\sqrt{a+b \rho^2 }\sqrt{a+ c \rho^2 }} \mathrm{d} x = \dfrac{\ln\left(\sqrt{b}\sqrt{c \rho ^2+a}+\sqrt{c}\sqrt{b \rho ^2+a}\right)}{\sqrt{b}\sqrt{c}}+ Const$, let $a=r^2$, $b=N_z^2 d_z^2$, and $c=N_y^2 d_y^2$, $\mu(r,0,0)$ can be represented as Eq. \eqref{prove_upa_mu3}.
	
	When it is UCPA structure, i.e., $N_y d_y=N_z d_z = D$, we have
	\begin{equation}
		\mu(r, 0,0)= \frac{\pi r^{2}}{D^2} \ln \left( \frac{D^2}{\pi r^2}+1 \right).
	\end{equation}
	
	2) When $\phi \neq 0$ or $\varphi \neq 0$, considering the UCPA structure, i.e., $N_y d_y=N_z d_z = D$,  $\mu(r,\phi,\varphi)$ can be calculated as Eq. \eqref{prove_ucpa1}.
	Since $\int \frac{1}{a \cos(\theta)+b \sin(\theta)+c} \mathrm{d} \theta = \frac{2}{\sqrt{c^2-a^2-b^2}} \arctan\left(\frac{(c-a)\tan(\theta/2)+b}{\sqrt{c^2-a^2-b^2}} \right) + Const$, let $a = -2 r \cos (\phi) \sin (\varphi) D \rho$, $b=-2 r \sin (\phi) D \rho$, and $c=D^{2} \rho^{2}+r^{2}$, $\mu(r,\phi,\varphi)$ can be translated into Eq. \eqref{prove_ucpa2}.
	Since $\int \frac{\rho}{\sqrt{a \rho^4+b \rho^2+c}} \mathrm{d} \rho = \frac{1}{4\sqrt{a}} \ln \left(\frac{2\sqrt{a} \sqrt{a \rho^4+b \rho^2+c} +(2a\rho^2+b)}{2\sqrt{a} \sqrt{a \rho^4+b \rho^2+c} -(2a\rho^2+b)} \right) + Const$, let $a = D^4$, $b=\left(4 \cos ^{2}(\phi) \cos ^{2}(\varphi)-2\right) r^{2} D^{2}$, and $c=r^4$, $\mu(r,\phi,\varphi)$ can be expressed as Eq. \eqref{prove_ucpa3}, where $\beta=\cos ^{2}(\phi) \cos ^{2}(\varphi)$.
	\begin{figure*}
		\begin{equation} \label{prove_upa_dao1}
			\begin{aligned}
				\frac{\partial \mu(r, \phi,\varphi)}{\partial r}
				=&\; \frac{\pi}{2D^2} \left(2r f_1(r,\beta) - \dfrac{4D^2 r}{{\pi} \sqrt{r^4+\frac{D^2 (4\beta-2) r^2}{{\pi}}+\frac{D^4}{{\pi}^2}} } \right)
			\end{aligned}
		\end{equation}
		\begin{equation} \label{prove_upa_dao2}
			\begin{aligned}
				\frac{\partial^2 \mu(r, \phi,\varphi)}{\partial r^2}
				=&\; \frac{\pi}{2D^2} \left(2 f_1(r,\beta) -  \dfrac{16D^2 }{{\pi}\sqrt{r^4+\frac{D^2 (4\beta-2) r^2}{{\pi}}+\frac{D^4}{{\pi}^2}}} +  \dfrac{4D^2 \left(3 \pi^2 r^4+2\pi D^2 (4\beta-2) r^2+D^4 \right)} {{\pi}^3  \left(r^4+\frac{D^2 (4\beta-2) r^2}{{\pi}}+\frac{D^4}{{\pi}^2}\right)^\frac{3}{2}} \right)
			\end{aligned}
		\end{equation}
		\begin{equation} \label{prove_upa_dao3}
			\begin{aligned}
				\frac{\partial^3 \mu(r, \phi,\varphi)}{\partial r^3}
				=&\; \frac{3 (4\beta-2) D^2 \pi^3 r^6+ 10 D^4 \pi^2 r^4+ (4\beta-2) D^6 \pi r^2-2D^8}{2D^2 \pi^4 \sqrt{r^4+\frac{D^2 (4\beta-2) r^2}{{\pi}}+\frac{D^4}{{\pi}^2}}}
			\end{aligned}
		\end{equation}
		\begin{equation} \label{prove_upa_f1}
			\begin{aligned}
				f_1 (r,\beta) = \ln \left[\frac{2 D^{2} \sqrt{\frac{D^{4}}{\pi^{2}}+\frac{\left(4 \beta -2\right) r^{2} D^{2}}{\pi}+r^{4}}+\left(\frac{2 D^{4}}{\pi}+\left(4 \beta -2\right) r^{2} D^{2}\right)}{2 D^{2} \sqrt{\frac{D^{4}}{\pi^{2}}+\frac{\left(4 \beta -2\right) r^{2} D^{2}}{\pi}+r^{4}}-\left(\frac{2 D^{4}}{\pi}+\left(4 \beta -2\right) r^{2} D^{2}\right)}\right]  +\ln \left[\frac{1- \beta }{ \beta }\right],
			\end{aligned}
		\end{equation}
	\end{figure*}
	
	\subsection{Proof of Proposition 2} \label{appendix_prop2}
	When $\phi=\varphi=0$, since $\frac{\ln(x+1)}{x}<1$ for $x>0$, $\mu(r,0,0)$ is always smaller than one.
	
	When $\phi \neq 0$ or $\varphi \neq 0$, the first, second and third partial derivative of $\mu(r,\beta)$ can be expressed as Eq. \eqref{prove_upa_dao1}, Eq. \eqref{prove_upa_dao2}, and Eq. \eqref{prove_upa_dao3}, respectively, where $f_1(r,\beta)$ is given in Eq. \eqref{prove_upa_f1}. 
	Let $f_2(r,\beta) =3 (4\beta-2) D^2 \pi^3 r^6+ 10 D^4 \pi^2 r^4+ (4\beta-2) D^6 \pi r^2-2D^8$, we have
	\begin{equation}
		\begin{aligned}
			\frac{\partial f_2(r,\beta)}{\partial r}
			=&\; 2{\pi}D^2r \left(9{\pi}^2 (4\beta-2) r^4+20{\pi}D^2r^2\right) \\
			& +  2{\pi}D^2r \left(D^4 (4\beta-2) \right) \\
			\triangleq & \; 2{\pi}D^2r f_3(r,\beta)
		\end{aligned}
	\end{equation}
	Let us consider the following three cases.
	
	1) When $\beta = \frac{1}{2}$, $\frac{\partial f_2(r,\beta)}{\partial r}
	= 40 {\pi}^2 D^4 r^2 >0$. Since $f_2(0,\beta)=-2D^8<0$ and $f_2(r\rightarrow +\infty,\beta)>0$, $f_2(r,\beta)=0$ only has one positive root.
	
	2) When $\frac{1}{2}<\beta <1 $, the discriminant of $f_3(r,\beta)$ is $\Delta = 4 \pi^2 D^4 \left( 100-9 (4\beta-2)^2 \right) > 0$ due to $ 0 \leq \beta<1$. The roots of $f_3(r,\beta)=0$ are $ \frac{D^2 \left(-10\pm \sqrt{100-9 (4\beta-2)^2 }\right)}{9\pi (4\beta-2)} <0$. Thus, $f_3(r,\beta)>0$, and $\frac{\partial f_2(r,\beta)}{\partial r}>0$, $f_2(r,\beta)$ only has one positive root.
	
	3) When $ 0< \beta <\frac{1}{2}$, the roots of $f_3(r,\beta)=0$ are $ \frac{D^2 \left(-10\pm \sqrt{100-9 (4\beta-2)^2 }\right)}{9\pi (4\beta-2)} >0$. Thus, $f_2(r,\beta)=0$  has two positive roots.
	
	Therefore, when $\frac{1}{2}\leq \beta <1 $, $\frac{\partial^3 \mu(r, \phi,\varphi)}{\partial r^3}=0$ only has one positive root, denoted by $r_1$. When $0<r<r_1$, $\frac{\partial^3 \mu(r, \phi,\varphi)}{\partial r^3}<0$; when $r>r_1$, $\frac{\partial^3 \mu(r, \phi,\varphi)}{\partial r^3}>0$. Note that $\frac{\partial^2 \mu(r, \phi,\varphi)}{\partial r^2}|_{r \rightarrow 0} \rightarrow +\infty$ and $\frac{\partial^2 \mu(r, \phi,\varphi)}{\partial r^2}|_{r \rightarrow +\infty}= 0$. Thus, $\frac{\partial^2 \mu(r, \phi,\varphi)}{\partial r^2}= 0$ has only one positive root, denoted by $r_2$. When $0<r<r_2$, $\frac{\partial^2 \mu(r, \phi,\varphi)}{\partial r^2}>0$; when $r>r_2$, $\frac{\partial^2 \mu(r, \phi,\varphi)}{\partial r^2}<0$. Note that $\frac{\partial \mu(r, \phi,\varphi)}{\partial r}|_{r \rightarrow 0} \rightarrow 0$ and $\frac{\partial \mu(r, \phi,\varphi)}{\partial r}|_{r \rightarrow +\infty}= 0$. Thereby, $\frac{\partial \mu(r, \phi,\varphi)}{\partial r}$ is always bigger than zeros. Hence, $\mu(r,\phi,\varphi)$ always increases as the distance $r$ increases, and $\mu(r,\phi,\varphi)$ is always smaller than one.
	
	When $ 0< \beta <\frac{1}{2}$, $\frac{\partial^3 \mu(r, \phi,\varphi)}{\partial r^3}=0$ has two positive roots, denoted by $r_3$ and $r_4$. When $0<r<r_3$ or $r>r_4$ , $\frac{\partial^3 \mu(r, \phi,\varphi)}{\partial r^3}<0$; when $r_3 <r< r_4 $, $\frac{\partial^3 \mu(r, \phi,\varphi)}{\partial r^3}>0$. As such, $\frac{\partial^2 \mu(r, \phi,\varphi)}{\partial r^2}= 0$ has two positive roots, denoted by $r_5$ and $r_6$. When $0<r<r_5$ or $r>r_6$ , $\frac{\partial^2 \mu(r, \phi,\varphi)}{\partial r^2}>0$; when $r_5 <r< r_6 $, $\frac{\partial^2 \mu(r, \phi,\varphi)}{\partial r^2}<0$. Then, $\frac{\partial \mu(r, \phi,\varphi)}{\partial r}=0$ has one positive root, denoted by $r_7$. When $0<r<r_7$, $\frac{\partial \mu(r, \phi,\varphi)}{\partial r}>0$; when $r> r_7 $, $\frac{\partial \mu(r, \phi,\varphi)}{\partial r}<0$. Therefore, in this case, $\mu(r,\phi,\varphi)$ will first increase, and then decrease until it approaches one as the distance increases. $\mu(r,\phi,\varphi)$ reaches its peak at the distance $r_7$. According to $f_3(r,\beta)=0$, we have
	\begin{equation}
		\begin{aligned}
			r_8 = D \sqrt{\frac{ 10+ \sqrt{100-9 (4\beta-2)^2 }}{9\pi (2-4\beta)}}.
		\end{aligned}
	\end{equation}
	It is found that $r_8$ is close to the point where the second partial derivative of $\mu(r,\phi,\varphi)$ over $r$ is zero. Then, $r_7$ is slightly smaller than $r_8$.

	\bibliographystyle{IEEEtran}
	\bibliography{near_field}

\begin{thebibliography}{10}
\providecommand{\url}[1]{#1}
\csname url@samestyle\endcsname
\providecommand{\newblock}{\relax}
\providecommand{\bibinfo}[2]{#2}
\providecommand{\BIBentrySTDinterwordspacing}{\spaceskip=0pt\relax}
\providecommand{\BIBentryALTinterwordstretchfactor}{4}
\providecommand{\BIBentryALTinterwordspacing}{\spaceskip=\fontdimen2\font plus
\BIBentryALTinterwordstretchfactor\fontdimen3\font minus
  \fontdimen4\font\relax}
\providecommand{\BIBforeignlanguage}[2]{{%
\expandafter\ifx\csname l@#1\endcsname\relax
\typeout{** WARNING: IEEEtran.bst: No hyphenation pattern has been}%
\typeout{** loaded for the language `#1'. Using the pattern for}%
\typeout{** the default language instead.}%
\else
\language=\csname l@#1\endcsname
\fi
#2}}
\providecommand{\BIBdecl}{\relax}
\BIBdecl

\bibitem{zhang2020prospective}
J.~Zhang, E.~Bj{\"o}rnson, M.~Matthaiou, D.~W.~K. Ng, H.~Yang, and D.~J. Love,
  ``Prospective multiple antenna technologies for beyond {5G},'' \emph{IEEE
  Journal on Selected Areas in Communications}, vol.~38, no.~8, pp. 1637--1660,
  Aug. 2020.

\bibitem{8869705}
W.~Saad, M.~Bennis, and M.~Chen, ``A vision of {6G} wireless systems:
  Applications, trends, technologies, and open research problems,'' \emph{IEEE
  Network}, vol.~34, no.~3, pp. 134--142, May 2020.

\bibitem{you2021towards}
X.~You, C.-X. Wang, J.~Huang, X.~Gao, Z.~Zhang, M.~Wang, Y.~Huang, C.~Zhang,
  Y.~Jiang, J.~Wang \emph{et~al.}, ``Towards {6G} wireless communication
  networks: Vision, enabling technologies, and new paradigm shifts,''
  \emph{Science China Information Sciences}, vol.~64, no.~1, pp. 1--74, 2021.

\bibitem{wang2022vision}
Z.~Wang, Y.~Du, K.~Wei, K.~Han, X.~Xu, G.~Wei, W.~Tong, P.~Zhu, J.~Ma, J.~Wang
  \emph{et~al.}, ``Vision, application scenarios, and key technology trends for
  {6G} mobile communications,'' \emph{Science China Information Sciences},
  vol.~65, no.~5, pp. 1--27, Mar. 2022.

\bibitem{9187980}
J.~C. Marinello, T.~Abr{\~a}o, A.~Amiri, E.~De~Carvalho, and P.~Popovski,
  ``Antenna selection for improving energy efficiency in xl-mimo systems,''
  \emph{IEEE Transactions on Vehicular Technology}, vol.~69, no.~11, pp.
  13\,305--13\,318, 2020.

\bibitem{akyildiz2016realizing}
I.~F. Akyildiz and J.~M. Jornet, ``Realizing ultra-massive {MIMO} (1024$\times$
  1024) communication in the (0.06--10) terahertz band,'' \emph{Nano
  Communication Networks}, vol.~8, pp. 46--54, 2016.

\bibitem{9170651}
E.~D. Carvalho, A.~Ali, A.~Amiri, M.~Angjelichinoski, and R.~W. Heath,
  ``Non-stationarities in extra-large-scale massive {MIMO},'' \emph{IEEE
  Wireless Communications}, vol.~27, no.~4, pp. 74--80, Aug. 2020.

\bibitem{8644126}
A.~Amiri, M.~Angjelichinoski, E.~de~Carvalho, and R.~W. Heath, ``Extremely
  large aperture massive {MIMO}: Low complexity receiver architectures,'' in
  \emph{2018 IEEE Globecom Workshops (GC Wkshps)}, Dec. 2018, pp. 1--6.

\bibitem{9617121}
H.~Lu and Y.~Zeng, ``Communicating with extremely large-scale array/surface:
  Unified modelling and performance analysis,'' \emph{IEEE Transactions on
  Wireless Communications}, pp. 1--1, Jun. 2021.

\bibitem{9580418}
A.~Liao, Z.~Gao, Y.~Yang, H.~H. Nguyen, H.~Wang, and H.~Yin, ``Angle estimation
  for terahertz ultra-massive {MIMO}-based space-to-air communications,'' in
  \emph{2021 IEEE/CIC International Conference on Communications in China
  (ICCC)}, 2021, pp. 776--781.

\bibitem{9672766}
S.~J. Maeng, Y.~Yapici, I.~Guvenc, A.~Bhuyan, and H.~Dai, ``Precoder design for
  physical-layer security and authentication in massive {MIMO} {UAV}
  communications,'' \emph{IEEE Transactions on Vehicular Technology}, vol.~71,
  no.~3, pp. 2949--2964, Mar. 2022.

\bibitem{8319526}
S.~Hu, F.~Rusek, and O.~Edfors, ``Beyond massive {MIMO}: The potential of data
  transmission with large intelligent surfaces,'' \emph{IEEE Transactions on
  Signal Processing}, vol.~66, no.~10, pp. 2746--2758, May 2018.

\bibitem{zhang2021wireless}
L.~Zhang, M.~Z. Chen, W.~Tang, J.~Y. Dai, L.~Miao, X.~Y. Zhou, S.~Jin,
  Q.~Cheng, and T.~J. Cui, ``A wireless communication scheme based on space-and
  frequency-division multiplexing using digital metasurfaces,'' \emph{Nature
  electronics}, vol.~4, no.~3, pp. 218--227, Mar. 2021.

\bibitem{9867922}
R.~Li, S.~Sun, Y.~Chen, C.~Han, and M.~Tao, ``Ergodic achievable rate analysis
  and optimization of {RIS}-assisted millimeter-wave {MIMO} communication
  systems,'' \emph{IEEE Transactions on Wireless Communications}, pp. 1--1,
  2022.

\bibitem{9724202}
X.~Shao, C.~You, W.~Ma, X.~Chen, and R.~Zhang, ``Target sensing with
  intelligent reflecting surface: Architecture and performance,'' \emph{IEEE
  Journal on Selected Areas in Communications}, vol.~40, no.~7, pp. 2070--2084,
  Jul. 2022.

\bibitem{9911191}
R.~Li, S.~Sun, and M.~Tao, ``Ergodic achievable rate maximization of
  {RIS}-assisted millimeter-wave {MIMO-OFDM} communication systems,''
  \emph{IEEE Transactions on Wireless Communications}, pp. 1--1, 2022.

\bibitem{s22145297}
\BIBentryALTinterwordspacing
S.~Sun and M.~Tao, ``Characteristics of channel eigenvalues and mutual coupling
  effects for holographic reconfigurable intelligent surfaces,''
  \emph{Sensors}, vol.~22, no.~14, Jul. 2022. [Online]. Available:
  \url{https://www.mdpi.com/1424-8220/22/14/5297}
\BIBentrySTDinterwordspacing

\bibitem{8736783}
B.~Friedlander, ``Localization of signals in the near-field of an antenna
  array,'' \emph{IEEE Transactions on Signal Processing}, vol.~67, no.~15, pp.
  3885--3893, Aug. 2019.

\bibitem{8453017}
Y.~Miao, J.-i. Takada, K.~Saito, K.~Haneda, A.~A. Glazunov, and Y.~Gong,
  ``Comparison of plane wave and spherical vector wave channel modeling for
  characterizing non-specular rough-surface wave scattering,'' \emph{IEEE
  Antennas and Wireless Propagation Letters}, vol.~17, no.~10, pp. 1847--1851,
  Oct. 2018.

\bibitem{8657705}
Y.~Miao, W.~Fan, J.~Takada, R.~He, X.~Yin, M.~Yang,
  J.~Rodr{\'i}guez-Pi{\"n}eiro, A.~A. Glazunov, W.~Wang, and Y.~Gong,
  ``Comparing channel emulation algorithms by using plane waves and spherical
  vector waves in multiprobe anechoic chamber setups,'' \emph{IEEE Transactions
  on Antennas and Propagation}, vol.~67, no.~6, pp. 4091--4103, Jun. 2019.

\bibitem{9091906}
Y.~Miao, K.~Haneda, J.-I. Naganawa, M.~Kim, and J.-I. Takada,
  ``Measurement-based analysis and modeling of multimode channel behaviors in
  spherical vector wave domain,'' \emph{IEEE Transactions on Wireless
  Communications}, vol.~19, no.~8, pp. 5345--5358, Aug. 2020.

\bibitem{balanis2012advanced}
C.~A. Balanis, \emph{Advanced engineering electromagnetics}.\hskip 1em plus
  0.5em minus 0.4em\relax John Wiley \& Sons, 2012.

\bibitem{balanis2015antenna}
------, \emph{Antenna theory: analysis and design}.\hskip 1em plus 0.5em minus
  0.4em\relax John wiley \& sons, 2015.

\bibitem{7942128}
K.~T. Selvan and R.~Janaswamy, ``Fraunhofer and fresnel distances : Unified
  derivation for aperture antennas.'' \emph{IEEE Antennas and Propagation
  Magazine}, vol.~59, no.~4, pp. 12--15, Aug. 2017.

\bibitem{1510955}
J.-S. Jiang and M.~Ingram, ``Spherical-wave model for short-range {MIMO},''
  \emph{IEEE Transactions on Communications}, vol.~53, no.~9, pp. 1534--1541,
  Sep. 2005.

\bibitem{4155681}
F.~Bohagen, P.~Orten, and G.~E. Oien, ``Design of optimal high-rank
  line-of-sight {MIMO} channels,'' \emph{IEEE Transactions on Wireless
  Communications}, vol.~6, no.~4, pp. 1420--1425, Apr. 2007.

\bibitem{4799060}
------, ``On spherical vs. plane wave modeling of line-of-sight {MIMO}
  channels,'' \emph{IEEE Transactions on Communications}, vol.~57, no.~3, pp.
  841--849, Mar. 2009.

\bibitem{6800118}
P.~Wang, Y.~Li, X.~Yuan, L.~Song, and B.~Vucetic, ``Tens of gigabits wireless
  communications over {E-Band} {LoS} {MIMO} channels with uniform linear
  antenna arrays,'' \emph{IEEE Transactions on Wireless Communications},
  vol.~13, no.~7, pp. 3791--3805, Jul. 2014.

\bibitem{7414041}
Z.~Zhou, X.~Gao, J.~Fang, and Z.~Chen, ``Spherical wave channel and analysis
  for large linear array in {LoS} conditions,'' in \emph{2015 IEEE Globecom
  Workshops (GC Wkshps)}, Dec. 2015, pp. 1--6.

\bibitem{liu2016channel}
L.~Liu, D.~W. Matolak, C.~Tao, Y.~Li, B.~Ai, and H.~Chen, ``Channel capacity
  investigation of a linear massive {MIMO} system using spherical wave model in
  {LOS} scenarios,'' \emph{Science China Information Sciences}, vol.~59, no.~2,
  pp. 1--15, Feb. 2016.

\bibitem{7501567}
J.~Chen, S.~Wang, and X.~Yin, ``A spherical-wavefront-based scatterer
  localization algorithm using large-scale antenna arrays,'' \emph{IEEE
  Communications Letters}, vol.~20, no.~9, pp. 1796--1799, 2016.

\bibitem{7981398}
X.~Yin, S.~Wang, N.~Zhang, and B.~Ai, ``Scatterer localization using
  large-scale antenna arrays based on a spherical wave-front parametric
  model,'' \emph{IEEE Transactions on Wireless Communications}, vol.~16,
  no.~10, pp. 6543--6556, Oct. 2017.

\bibitem{lu2021does}
H.~Lu and Y.~Zeng, ``How does performance scale with antenna number for
  extremely large-scale {MIMO}?'' in \emph{ICC 2021-IEEE International
  Conference on Communications}.\hskip 1em plus 0.5em minus 0.4em\relax IEEE,
  Jun. 2021, pp. 1--6.

\bibitem{cui2021near}
\BIBentryALTinterwordspacing
M.~Cui, L.~Dai, R.~Schober, and L.~Hanzo, ``Near-field wideband beamforming for
  extremely large antenna array,'' \emph{arXiv preprint arXiv:2109.10054},
  2021. [Online]. Available: \url{https://arxiv.org/abs/2109.10054}
\BIBentrySTDinterwordspacing

\bibitem{bjornson2021primer}
E.~Bj{\"o}rnson, {\"O}.~T. Demir, and L.~Sanguinetti, ``A primer on near-field
  beamforming for arrays and reconfigurable intelligent surfaces,'' in
  \emph{2021 55th Asilomar Conference on Signals, Systems, and Computers},
  2021, pp. 105--112.

\bibitem{9538864}
C.~Feng, H.~Lu, Y.~Zeng, S.~Jin, and R.~Zhang, ``Wireless communication with
  extremely large-scale intelligent reflecting surface,'' in \emph{2021
  IEEE/CIC International Conference on Communications in China (ICCC
  Workshops)}, Jul. 2021, pp. 165--170.

\bibitem{9133126}
J.~C.~B. Garcia, A.~Sibille, and M.~Kamoun, ``Reconfigurable intelligent
  surfaces: Bridging the gap between scattering and reflection,'' \emph{IEEE
  Journal on Selected Areas in Communications}, vol.~38, no.~11, pp.
  2538--2547, Nov. 2020.

\bibitem{zhang2021intelligent}
\BIBentryALTinterwordspacing
M.~Zhang and X.~Yuan, ``Intelligent reflecting surface aided {MIMO} with
  cascaded line-of-sight links: Channel modelling and capacity analysis,''
  \emph{arXiv preprint arXiv:2109.08913}, 2021. [Online]. Available:
  \url{https://arxiv.org/abs/2109.08913}
\BIBentrySTDinterwordspacing

\bibitem{cui2022near}
\BIBentryALTinterwordspacing
M.~Cui, Z.~Wu, Y.~Lu, X.~Wei, and L.~Dai, ``Near-field communications for {6G}:
  Fundamentals, challenges, potentials, and future directions,'' \emph{arXiv
  preprint arXiv:2203.16318}, 2022. [Online]. Available:
  \url{https://arxiv.org/abs/2203.16318}
\BIBentrySTDinterwordspacing

\bibitem{7098875}
O.~Roy and M.~Vetterli, ``The effective rank: A measure of effective
  dimensionality,'' in \emph{2007 15th European Signal Processing Conference},
  Sep. 2007, pp. 606--610.

\bibitem{sun2018propagation}
S.~Sun, T.~S. Rappaport, M.~Shafi, P.~Tang, J.~Zhang, and P.~J. Smith,
  ``Propagation models and performance evaluation for {5G} millimeter-wave
  bands,'' \emph{IEEE Transactions on Vehicular Technology}, vol.~67, no.~9,
  pp. 8422--8439, Sep. 2018.

\end{thebibliography}

	\end{document}